\documentclass[final,1p,times]{elsarticle}

\usepackage{amsmath,amssymb,amsthm,amssymb,amsfonts}

\usepackage{url}
\usepackage{bm}
\usepackage{hyperref,cleveref}
\usepackage[ruled,linesnumbered,boxed,boxruled,vlined]{algorithm2e}
\usepackage{todonotes}

\newtheorem{theorem}{Theorem}[section]

\newtheorem{lemma}[theorem]{Lemma}

\newcommand{\OB}{\mathcal{O}_B}
\newcommand{\OO}{\mathcal{O}}
\newcommand{\sOB}{\widetilde{\mathcal{O}}_B}
\newcommand{\sOO}{\widetilde{\mathcal{O}}}

\def\norm#1{\mathopen\| #1 \mathclose\|}

\def\vol{\mbox{vol}}

\newcommand{\Z}{\mathbb{Z}}

\newcommand{\R}{\mathbb{R}}

\newcommand{\cC}{\mathcal{C}}

\newcommand{\lmi}[1]{\bm{F}(\bm{#1})}
\renewcommand{\line}[2]{\ell(\bm{#1}, \bm{#2})}
\newcommand{\transpose}[1]{\bm{#1}^\top}
\newcommand{\adjoint}[1]{\bm{#1}^*}
\newcommand{\trace}[1]{\mathtt{Trace}\left(#1\right)}
\newcommand{\segment}[2]{[\bm{#1}, \bm{#2}]}

\newcommand{\boundary}[1]{{\partial #1}\xspace}
\newcommand{\interior}[1]{#1^{\circ}\xspace}
\newcommand{\rank}[1]{\mathtt{rank}(#1)\xspace}

\newcommand{\intersectionop}{\textsc{intersection}\xspace}
\newcommand{\reflectionop}{\textsc{reflection}\xspace}
\newcommand{\membershipop}{\textsc{membership}\xspace}
\newcommand{\PEP}{\textsc{pep}\xspace}

\newcommand{\unitball}{\mathcal{B}_n}
\newcommand{\pickrandom}{\leftarrow_R}

\newcommand{\billiard}{\textsc{W-Billard}\xspace}
\newcommand{\hitandrun}{\textsc{W-HnR}\xspace}
\newcommand{\cdhr}{\textsc{W-CHnR}\xspace}
\newcommand{\hmcr}{\textsc{W-HMC-r}\xspace}

\journal{...}

\begin{document}

\begin{frontmatter}

  
\title{Efficient Sampling  from Feasible Sets of SDPs and Volume Approximation}


\author[ATH,DI]{Apostolos Chalkis}
\author[ATH,DI]{Ioannis Z. Emiris}
\author[DI]{Vissarion Fisikopoulos}
\author[DI]{Panagiotis Repouskos}
\author[FR]{Elias Tsigaridas}

\address[ATH]{ATHENA Research Center, Greece}
\address[DI]{Department of Informatics \& Telecommunications \\ National \& Kapodistrian University of Athens, Greece}
\address[FR]{Inria Paris, IMJ-PRG,\\ Sorbonne Universit\'e and Paris Universit\'e}

\begin{abstract}
  We present algorithmic, complexity, and implementation results on
  the problem of sampling points from a
  spectrahedron, that is the feasible region of a semidefinite program.
 
  Our main tool is geometric random walks. We analyze the arithmetic and bit
  complexity of certain primitive geometric operations that are based on the
  algebraic properties of spectrahedra and the polynomial eigenvalue problem.
This study leads to the implementation of a broad collection of random walks for
sampling from spectrahedra that experimentally show faster mixing times than methods
currently employed either in theoretical studies or in applications, 
  including the popular family of Hit-and-Run walks.
The different random walks offer a variety of advantages, 
thus allowing us to efficiently sample from general probability distributions, for example the
  family of log-concave distributions which arise in numerous applications.  
  We focus on two major applications of independent interest:
  (i) approximate the volume of a spectrahedron, and
  (ii) compute the expectation of functions coming from robust optimal control. 

We exploit efficient linear algebra algorithms and
  implementations to address the aforementioned computations in very high dimension. 
In particular, we provide a C++ open source implementation of our methods
  that scales efficiently, for the first time, up to dimension~200.
  We illustrate its efficiency on various data sets.
\end{abstract}



\begin{keyword}
spectahedra \sep semidefinite-programming \sep sampling \sep random walk \sep polynomial eigenvalue problem \sep volume approximation \sep optimization



\end{keyword}

\end{frontmatter}

\section{Introduction}

Spectrahedra are probably the most well studied shapes after
polyhedra.  We can represent polyhedra as the intersection of the
positive orthant with an affine subspace. Spectrahedra  generalize polyhedra,
in the sense that they are the intersection of the cone of positive
semidefinite matrices ---{\it i.e.,}\ symmetric matrices with
non-negative eigenvalues--- with an affine space.
In other words, a spectrahedron $S \subset \R^n$ is the feasible set
of a linear matrix inequality. That is, let
\begin{equation} \label{eq:lmi}
\lmi{x} = \bm{A}_0 + x_1 \bm{A}_1 + \cdots + x_n \bm{A}_n  ,
\end{equation}
where $\bm{A}_i$ are symmetric matrices in $\R^{m \times m}$, then
$S = \{ \bm{x} \in \R^n \,|\, \lmi{x} \succeq 0 \}$, where $\succeq$ denotes
positive semidefiniteness. We assume throughout that $S$ is bounded of dimension
$n$. Spectrahedra are convex sets (Figures~\ref{fig:spectra-curve} and~\ref{fig:sampling}) and every polytope is a spectrahedron, but not the opposite.
They are the feasible regions of semidefinite programs~\cite{Ramana99} in the
way that polyhedra are feasible regions of linear programs.

\begin{figure}[t]
	\begin{center}
		\includegraphics[scale=0.10]{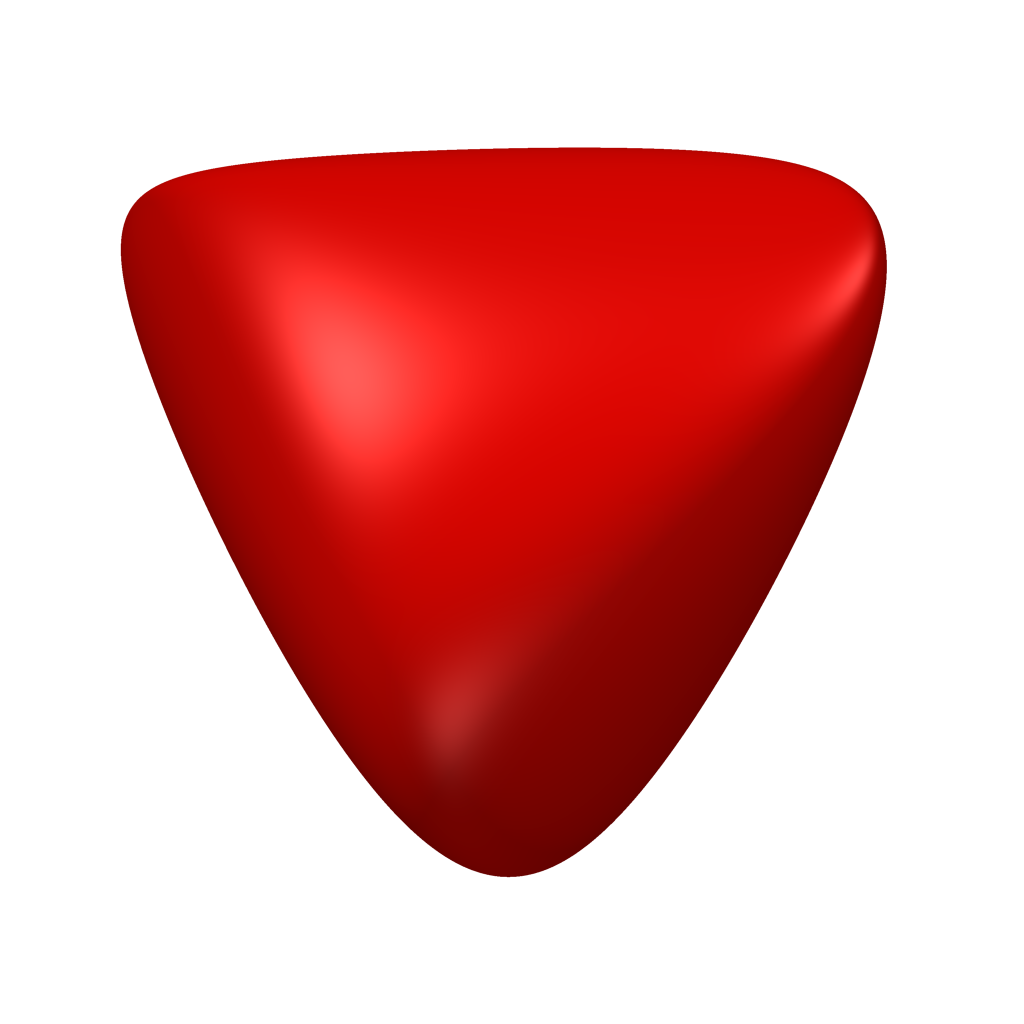}
		\hspace{1.5cm}
		\includegraphics[scale=0.43]{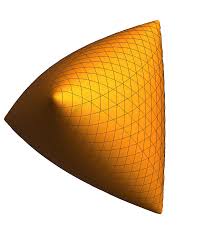}
		\caption{Left, a 3D elliptope (image from \url{https://en.wikipedia.org/wiki/Spectrahedron}); an elliptope --a special case of spectrahedron-- is the set of all real, square symmetric matrices whose diagonal entries are all equal to one and whose eigenvalues are all non-negative. Right, a spectrahedron from \cite{Breiding17}; it is called cubic spectrahedron.}
		\label{fig:elliptope}
	\end{center}
\end{figure}


Efficient methods for sampling points in spectrahedra are crucial for
many applications, such as volume approximation~\cite{Cousins15},
integration~\cite{Lovasz06}, semidefinite
optimization~\cite{Lovasz06,Kalai06}, and applications in robust
control
analysis~\cite{CalCam-rc-cdc-03,Calafiore-cdc-04,TCD-rc-bk-12}. To
sample in the interior or on the boundary of $S$, we employ geometric
random walks~\cite{Vempala05}. A geometric random walk on $S$ starts
at some interior point and at each step moves to a "neighboring" point
that we choose according to some distribution, depending only on the
current point; thus it is a special category of Markov chains. For example, in
the so-called ball walk, we move to a point $p$ that we choose uniformly at
random in a ball of fixed radius $\delta$, if $p\in S$. The complexity of a
random walk depends on its mixing time ---the number of steps required to bound
the distance between the current and the stationary distribution--- and the
complexity of the basic geometric operations performed at each step of the walk;
we call the latter per-step complexity.

The majority of geometric random walks are defined for general convex bodies and
are based on an oracle; usually the membership oracle. There are also a few walks,
e.g., Vaidya walk~\cite{Chen17} and the sub-linear ball walk~\cite{Mangoubi19},
specialized for polytopes. Most results on their analysis focus on convergence
and mixing time, while they define abstractly the operations they perform at
each step and they enclose them in the corresponding oracle. That is why the
complexity bounds involve the number of oracle calls.

To specialize a random walk for a family or representation of convex
bodies one has to come up with efficient algorithms for the basic
geometric operations to realize the (various) oracles. These
operations should exploit the underlying geometric and algebraic
properties and are of independent interest.
They depend on efficient (numerical) linear algebra computations.
More importantly, they
dominate the per-step complexity and are hence crucial both for the
overall complexity to sample a point from the target distribution,
as well as for an efficient implementation.

The study of basic geometric operations to sample from non-linear convex objects
finds its roots in non-linear computational geometry. During the last two
decades, there are combined efforts~\cite{BT-nonlinear-book-06,EmSoTh} to develop
efficient algorithms for the basic operations (predicates) that are behind
classical geometric algorithms, like convex hull, arrangements, Voronoi
diagrams, to go beyond points and lines and handle curved objects. For this, one
exploits efficiently structure and symmetry in linear algebra computations, and
develops novel algebraic tools.  

To our knowledge, only the {\em Random Directions Hit and Run} (\hitandrun) random walk~\cite{Smith84} has been
studied for spectrahedra \cite{Calafiore-cdc-04}. To exploit the various other
walks, like {\em Ball walk} \cite{Vempala05}, {\em Billiard walk} (\billiard)~\cite{GryazinaBilliard},
and {\em Hamiltonian Monte Carlo with boundary reflections} (\hmcr)~\cite{Mohasel15},
one needs to provide geometric
operations. For example, we need to compute the reflection of a curve at the
boundary, and the intersection point of a curve with the boundary (of a
spectrahedron).

We should mention that there is a gap \cite{Cousins16,Betancourt17} between the
theoretical worst case bounds for the mixing times and the practical
performance of the random walk algorithms.  
Furthermore, there are random walks without known
theoretical mixing times, such as Coordinate Directions Hit and Run (\cdhr), \billiard\ and \hmcr. To study them experimentally,
it is imperative to provide an efficient realization of the corresponding oracles.





\paragraph{Previous Work}
Sampling convex sets via random walks has attracted a lot of interest in the last
decades. Most of the works assume that the convex sets are polytopes;
\cite{Mangoubi19} provides an overview of the state-of-the-art. Random walks on
spectrahedra are studied in~\cite{Polyakdecomposition, Dabbene2010ARC}, where
they exploit the \hitandrun\ and the computation of the intersection of the
walk with the boundary reduces to a generalized eigenvalue~problem.

The \billiard~\cite{GryazinaBilliard} is a general way of
sampling in convex or non-convex shapes from the uniform distribution.  A
mathematical billiard consists of a domain $\mathcal{D}$ and a
point-mass, moving freely in $\mathcal{D}$
\cite{Tabachnikov2005GeometryAB}. When this point-mass hits the
boundary, it performs a specular reflection, albeit without losing
velocity. An  application of billiards is the study of optical
properties of conics \cite[Section~4]{Tabachnikov2005GeometryAB}.

If the trajectory is not a line, but rather a parametric curve,
then the intersection operation reduces to solving a polynomial eigenvalue problem (PEP);
\hmcr\ requires this operation. 
PEP is a well-studied  problem in computational mathematics, e.g.,
\cite{Tis-becpep-2000}, and it appears in many applications.
There are important results both for the perturbation analysis of PEP 
\cite{Tis-becpep-2000,Berhanu-phd-05,DedTis-pep-pert-03},
as well as for the 
condition-based analysis of algorithms for the real and complex versions of PEP,
if we assume  random inputs 
\cite{ArmBel-rand-pep-19,BelKoz-real-pep-19}.

For the closely related problem of volume computation, there is also an
extensive bibliography~\cite{Bueler98}. The bulk of the theoretical
studies are either for general convex bodies~\cite{Dyer91, Cousins15} or polytopes~\cite{LeeRiem18}.
Practical algorithms and implementations exist only for polytopes~\cite{Emiris18, Cousins16}. 
Nevertheless, there are notable exceptions that consider algorithms
for computing the volume of compact (basic) semi-algebraic sets.  For
example, in \cite{LMS-sa-volume-19} they exploit the periods of rational
integrals. In the same setting,
\cite{KorHen-mom-sa-volume-18,HLS-sa-volume-09} introduce numerical
approximation schemes for volume computations, which rely on the
moment-based algorithms and semi-definite programming.

Finally, sampling from a multivariate distribution is a central
problem in numerous applications. For example, it is useful in robust
control analysis \cite{Calafiore-cdc-04,CalCam-rc-cdc-03,TCD-rc-bk-12} to overcome
the worst case hardness as well as in integration~\cite{Lovasz06} and
convex optimization
\cite{Dabbene2010ARC,Kalai06}.

\paragraph{Our contribution}
We present a framework of basic geometric operations for computations with
spectrahedra. In particular, we analyze the arithmetic and bit complexity of the three fundamental
operations, \membershipop, \intersectionop, \reflectionop, by reducing them to
 linear algebra computations.
Based on this framework, we support a rich class of
geometric random walks, which in turn we employ in order to build efficient
methods for sampling points from spectahedra, under various
probability distributions. We apply these tools to approximate the
volume of spectrahedra, as well as to integrate over spectrahedral domains.
This extends the limits of the state-of-the-art methods that sample
from spectrahedra, which actually involve only the \hitandrun~\cite{Dabbene2010ARC,Calafiore-cdc-04}.
We offer an efficient
C++ implementation of our algorithms as a development branch of the package {\tt
	volesti}\footnote{\scriptsize\url{https://github.com/GeomScale/volume_approximation/tree/sample_spectrahedra}}, an open source library for high dimensional sampling and volume
computation.
While the implementation is in C++, there is also an {\tt R} interface, for
easier access to its functionality. Our implementation is based on state-of-the-art algorithms in numerical linear algebra to address computation in high dimension. Our software makes use of powerful C++ libraries such as {\tt Eigen}~\cite{eigenlib} and {\tt Spectra}~\cite{Spectra}.

We demonstrate the efficiency of our approach and implementation on
problems from robust control  and optimization.
First, as a special case of integration, we approximate the volume of
spectrahedra up to dimension 100; this is, to the best of our
knowledge, the first time such computations for non-linear objects are
performed in high dimensions.  Then, we approximate the expected value
of a function $f:\R^n\rightarrow [0,1]$, whose argument is a random
variable having uniform distribution over a spectrahedron of dimension~$200$.

Finally, we sample from the Boltzmann distribution using \hmcr; this
exploits a random walk in a spectrahedron that employs a
polynomial trajectory of degree two. Sampling using \hmcr\ from
truncated distributions is a classical problem in computational
statistics~ \cite{Chevallier18, Mohasel15}; alas, existing approaches
handle either special distributions or special cases of constraints
\cite{Pakman14, Lan14}.
We equip \hmcr\ with geometric operations to handle log-concave densities
truncated by linear matrix inequalities (LMI) constraints. 
A combination of Boltzmann distribution with a
simulated annealing technique~\cite{Kalai06} could lead to a
practical solver for semidefinite programs (SDP).  

A very short version of this paper has appeared as a poster in \cite{issac_poster}.


\paragraph{Paper organization} First, we introduce our notation.
In Section~\ref{sec:geom-ops} we introduce the basic geometric operations 
used to efficiently implement membership and boundary oracles.
In Section~\ref{sec:rand-walks}, we develop the different types of random
walks. Finally, Section~\ref{sec:implementation} presents our
implementation, our applications, and various experiments.

\paragraph{Notation}
We denote by $\OO$, respectively $\OB$, the arithmetic, respectively  bit, complexity
and we use $\sOO$, respectively $\sOB$, to ignore (poly-)logarithmic
factors.  The bitsize of a univariate polynomial $A \in \Z[x]$ is the
maximum bitsize of its coefficients.
We use bold letters for matrices, $\bm{A}$, and vectors, $\bm{v}$; we
denote by $A_{i,j}$, respectively $v_i$, their elements;
$\transpose{A}$ is the transpose and $\adjoint{A}$ the adjoint of
$\bm{A}$.  If $\bm{x} = (x_1, \dots, x_n)$, then
$\lmi{x} = \bm{A}_0 + \sum_{i=1}^n x_i \bm{A}_i$, see 
(\ref{eq:lmi}).
For two points $\bm{x}$ and $\bm{y}$, we denote the line through them
by $\line{x}{y}$ and their segment as 
$\segment{x}{y}$.
For a spectrahedron $S$, let  its interior be $\interior{S}$
and its boundary $\boundary{S}$.
We represent a probability distribution $\pi$ with a probability density function $\pi(x)$. When $\pi$ is truncated to $S$ the support of $\pi(x)$ is $S$.
If $\pi$ is log-concave,  then $\pi(\bm{x})\propto e^{-\alpha f(\bm{x})}$, where $f:\R^d\rightarrow \R$ a convex function.
Finally, let $\unitball$ be the $n$-dimensional unit ball and denote by
$\boundary{\unitball}$ its boundary.

\section{Basic geometric operations} \label{sec:geom-ops}	

Our toolbox for computations with spectrahedra and implementing random
walks, consists of three basic geometric operations: \membershipop,
\intersectionop, and \reflectionop.  

For a spectrahedron $S$,
\membershipop decides whether a point lies inside $S$, \intersectionop
computes the intersection of an algebraic curved trajectory
$\mathcal{C}$ with the boundary $\boundary{S}$, and \reflectionop computes the
reflection of an algebraic curved trajectory when it hits
$\boundary{S}$.  The last two operations are required because random walks
can move along non-linear trajectories inside convex bodies.  For those
studied in this paper, the trajectories are parametric polynomial
curves, of various degrees.  Computation with these curves reduces to
solving the polynomial eigenvalue problem (PEP).  


\subsection{Analysis of  PEP} \label{sec:pep}

To estimate the complexity of \intersectionop we need to bound the complexity of PEP.
The Polynomial Eigenvalue Problem (PEP) consists in computing
$\lambda \in \R$ and $\bm{x} \in \R^m$ that satisfy the (matrix)
equation
\begin{equation}
\label{eq:PEP}
P(\lambda) \, \bm{x} = 0 \Leftrightarrow
(\bm{B}_d \lambda^d + \cdots + \bm{B}_1 \lambda + \bm{B}_0) \bm{x} = 0
\enspace ,
\end{equation}
where $P(\lambda)$ is a univariate matrix polynomial whose coefficients are
matrices $\bm{B}_i \in \R^{m \times m}$. We further assume that
$\bm{B}_d$ and $\bm{B}_0$ are invertible. In general, there are
$\delta = m \, d$ values of $\lambda$.  We refer the reader to
\cite{Tis-becpep-2000} for a thorough exposition of PEP.

One approach for solving PEP is to linearalize the problem and to
express the $\lambda$'s as the eigenvalues of a larger matrix.  For this,
we transform Equation~(\ref{eq:PEP}) into a linear pencil of dimension $\delta$.
Following~\cite{Berhanu-phd-05}, the {\em Companion
	Linearization} consists in solving the generalized eigenvalue
problem $\bm{C}_0 - \lambda \bm{C}_1$, where
{ 
  \[
	\bm{C}_0=
	\left[
	\begin{array}{cccc}
	\bm{B}_d & 0 & \cdots & 0 \\
	0 & \bm{I}_m &  \ddots & \vdots \\
	\vdots & \ddots & \ddots & 0\\
	0 & \cdots & 0 & \bm{I}_m
	\end{array}
	\right]
	\text{ and }
	\bm{C}_1 =
	\left[
	\begin{array}{cccc}
	\bm{B}_{d-1} & \bm{B}_{d-2} & \cdots & \bm{B}_0 \\
	-\bm{I}_m & 0 &  \cdots & 0 \\
	\vdots & \ddots & \ddots & \vdots\\
	0 & \cdots & -\bm{I}_m & 0
	\end{array}
	\right] ,
	\]
}
\noindent
and $\bm{I}_{m}$ denotes the $m \times m$ identity matrix.
The eigenvectors $\bm{x}$ and $\bm{z}$ are related as follows:
$\bm{z} = [1, \lambda, \dots, \lambda^{d-2}, \lambda^{d-1}]^{\top}
\otimes \bm{x}$.

To obtain an exact algorithm for PEP we exploit the assumption that
$\bm{B}_d$ is invertible so as to transform the problem to the following
classical eigenvalue problem
$(\lambda \bm{I}_d - \bm{C}_2) \bm{z} = 0$, where
\[
\bm{C}_2 =
\left[
\begin{array}{cccc}
\bm{B}_{d-1}\bm{B}_d^{-1} & \bm{B}_{d-2}\bm{B}_d^{-1} & \cdots & \bm{B}_0\bm{B}_d^{-1} \\
-\bm{I}_m & 0 &  \cdots & 0 \\
\vdots & \ddots & \ddots & \vdots\\
0 & \cdots & -\bm{I}_m & 0
\end{array}
\right] .
\]
The eigenvectors are roots of the characteristic polynomial of
$\bm{C}_2$.
Now the problem is to compute the eigenvalues of
$\bm{C} \in \R^{\delta \times \delta}$.
From a complexity point of view, the best algorithm 
relies on computing the roots of the characteristic
polynomial \cite{Schonhage-usolve}, which is the method we follow in order
to derive complexity bounds. 
Of course, in practice other methods are faster or more stable.  

\begin{lemma}
	\label{lem:pep}
	Consider a PEP of degree $d$, involving matrices of dimension
	$m \times m$, with integer elements of bitsize at most $\tau$, see 
	Equation~(\ref{eq:PEP}).  There is a randomized algorithm for computing 			the eigenvalues and the eigenvectors of PEP up to precision
	$\epsilon = 2^{-L}$, in $\sOB(\delta^{\omega+3} L)$, where $\delta = m d$
	and $L = \Omega(\delta^3 \tau)$.
	The arithmetic complexity  is
	$\sOO( \delta^{2.697} + \delta\lg(1/\epsilon))$.
\end{lemma}
\begin{proof}
	We can compute the characteristic polynomial of an $N \times N$ matrix $M$
	in $\sOB( N^{2.697 + 1}\lg\norm{M})$ using a randomized algorithm, see
	\cite{KalVil-det-05} and references therein.  Here $\norm{M}$ denotes
	the largest entry in absolute value.
	In our case, the elements of $\bm{C}_2$ have bitsize $\sOO(\delta\tau)$
	and its characteristic polynomial is 
	of degree $d$ and coefficient bitsize $\sOB( \delta^2 \tau)$.  We compute it in
	$\sOB(\delta^{2.697+1} \delta\tau) = \sOB(\delta^{4.697} \tau)$ and isolate all its
	real roots in $\sOB(\delta^5 + \delta^4\tau)$ \cite{Pan-opt-usolve-02};
	they correspond to the real eigenvalues of PEP. We can decrease the width of
	the isolating interval by a factor of $\epsilon = 2^{-L}$ for
	\emph{all} the roots in $\sOB(\delta^3\tau + \delta L)$ \cite{PanTsi-refine-16}.
	Thus, the overall complexity is $\sOB(\delta^5 + \delta^{4.697}\tau + \delta L)$.
	
	It remains to compute the corresponding eigenvectors.  For each
	eigenvalue $\lambda$ we may compute the corresponding eigenvector
	$\bm{z}$ by Gaussian elimination and back substitution to
	the (augmented) matrix $[ \lambda \bm{I}_{\delta} - \bm{C}_2 \,|\, \bm{0} ]$.
	This requires $\sOO(\delta^{\omega})$ arithmetic operations.  However,
	as $\lambda$ is a root of the characteristic polynomial, one has to
	operate on algebraic numbers, which is highly non-trivial,
and one needs to bound the number of bits needed to
	compute the elements of $\bm{z}$ correctly and to recover $\bm{x}$.

	Hence, we employ separation bounds for polynomial system
	adapted to eigenvector computation \cite{emt-dmm-j}.
	One needs, as in the case of eigenvalues, $\sOB(\delta^4 + \delta^3\tau)$ bits to
	isolate the coordinates of the eigenvectors.
	To decrease the width of
	the corresponding isolating intervals by a factor of $\epsilon = 2^{-L}$,
	the number of bits becomes $\sOB(\delta^4 + \delta^3\tau + L)$.
	Thus, we compute the eigenvectors in 
	$\sOB(\delta^{\omega}(\delta^4 + \delta^3\tau + L)) =
	\sOB(\delta^{\omega + 4} + \delta^{\omega + 3}\tau + \delta^{\omega} L)$.

	For the arithmetic complexity we proceed as follows:  We compute
	the characteristic polynomial in $\sOO(\delta^{2.697})$, we approximate its
	roots up to $\epsilon$ in $\sOO(\delta \lg(1/\epsilon))$.  Finally, we
	compute the eigenvectors with $\sOO(\delta^{\omega})$ arithmetic operations.
	So the overall cost is $\sOO(\delta^{2.697} + \delta \lg(1/\epsilon))$.
\end{proof}



\subsection{\membershipop} \label{sec:membership}

The operation $\membershipop(\bm{F}, \bm{p})$ decides whether a point
$\bm{p}$ lies in the interior of a spectrahedron
$S = \{ \bm{x} \in \R^n \,|\, \lmi{x} \succeq 0 \}$.  For this, first,
we construct the matrix $\lmi{\bm{p}}$.  If it is
positive definite, then $\bm{p} \in \interior{S}$.
If it is positive
semidefinite, then $\bm{p} \in \boundary{S}$, otherwise
$\bm{p} \in \R^n \setminus S$.  The pseudo-code appears in
Algorithm~\ref{alg:membershipop}.

\begin{algorithm}
	\caption{$\membershipop(\bm{F}, \bm{p})$}
	\label{alg:membershipop}
	\SetKwInOut{Input}{Input}
	\SetKwInOut{Output}{Output} 
	
	\Input{An LMI $\lmi{x} \succeq 0$ representing  a spectrahedron $S$
		and a point $\bm{p} \in \R^n$.}
	
	\Output{\textsc{true} if $\bm{p}\in S$, \textsc{false} otherwise.}
	
	\BlankLine
	
	$\lambda_{min} \leftarrow$ smallest eigenvalue of $\lmi{p}$\;
	\lIf{$\lambda_{min} \geq 0$} {
		\KwRet \textsc{true} 
	}
	\KwRet \textsc{false} \;  
\end{algorithm}

\begin{lemma}\label{lem:membership}
	Algorithm~\ref{alg:membershipop}, $\membershipop(\bm{F}, \bm{p})$, 
	requires $\sOO(n m^2 + m^{2.697})$ arithmetic operations.  If 
	$\bm{F}$ and $\bm{p}$ have integer elements of bitsize at most
	$\tau$, respectively $\sigma$, then the bit complexity is
	$\sOB((n m^2 + m^{3.697})(\tau + \sigma))$.
\end{lemma}
\begin{proof}
	We construct $\bm{F}(\bm{p})$ in $\OO(n m^2)$. Then, with
	$\OO(m^{2.697})$ operations we compute its characteristic polynomial
	\cite{KalVil-det-05} and in $\sOO(m)$ we decide whether it has negative
	roots, for example by solving \cite{Pan-opt-usolve-02} or using fast
	subresultant algorithms \cite{Lecerf-subres-2019,LicRoy-sub-res-01}.
	For the bit complexity, the construction costs
	$\sOB(n m^2 (\tau + \sigma))$ and the computation of the characteristic
	polynomial costs $\sOB(m^{2.697+1}(\tau + \sigma))$, using a randomized
	algorithm \cite{KalVil-det-05}.  One may test for negative roots, and
	thus eigenvalues, in $\sOB(m^2n (\tau + \sigma))$
	\cite{LicRoy-sub-res-01}.
\end{proof}

\subsection{\intersectionop} \label{sec:interection}

Consider a parametric polynomial curve
$\mathcal{C}$ such that it has a non-empty intersection with a
spectrahedron $S$.  Throughout, we consider only the real trace of $\cC$. Assume that the value of the parameter $t =
0$ corresponds to a point $\bm{p}_0$, that lies in $\mathcal{C} \cap
\interior{S}$.  Furthermore, assume that the segment
of $\mathcal{C}$ on which $\bm{p}_0$ lies, intersects the boundary of
$S$ transversally at two points, say $\bm{p}_{-}$ and
$\bm{p}_{+}$.  The operation
$\intersectionop$ computes the parameters, $t_{-}$ and
$t_{+}$, corresponding to these two points.
Figure~\ref{fig:spectra-curve} illustrates this discussion and the
pseudo-code of \intersectionop appears in
Algorithm~\ref{alg:intersectionop}.

\begin{figure}[t]
	\begin{center}
		\includegraphics[scale=0.20]{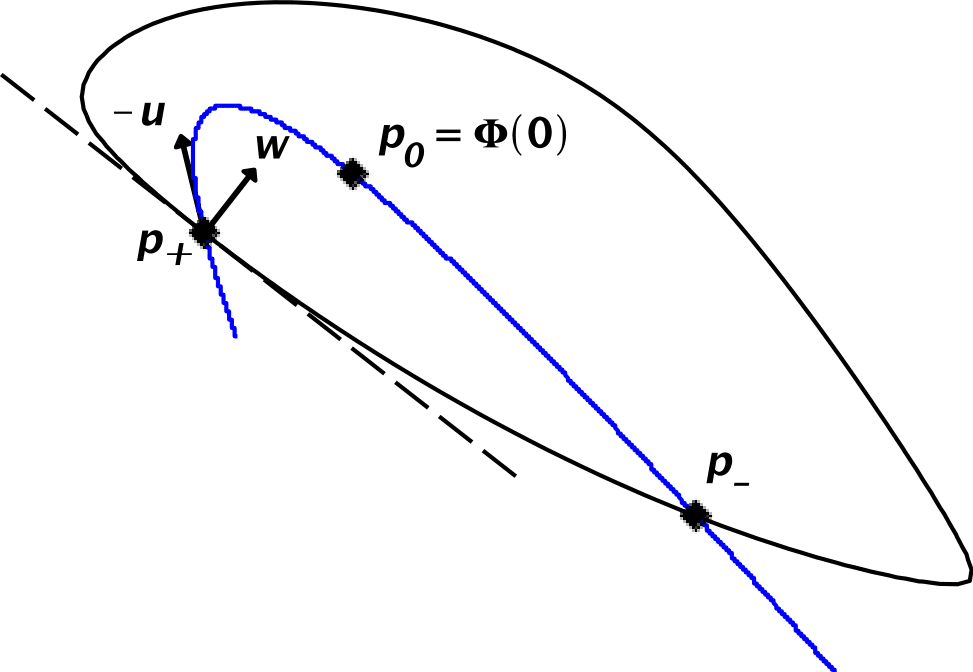}
		\caption{\emph{\textmd{A spectrahedron $S$ described by $\lmi{x}$ and a parameterized curve $\Phi$. The point $\bm{p}_0=\Phi(0)$ lies in the interior of $S$, and the points $\bm{p}_+ = \Phi(t_+)$ and $\bm{p}_-=\Phi(t_-)$ on the boundary. Vector $\bm{w}$ is the surface normal of $\boundary{S}$ at $\bm{p}_+$ and $\bm{u}$ is the direction of $\Phi$ at time $t=t_+$.}}
		\label{fig:spectra-curve}}
	\end{center}
\end{figure}



To prove correctness  and estimate the complexity we proceed as
follows: As before (see also Equation~\ref{eq:lmi}),  $S$ is the feasible region of linear matrix inequalities (LMI)
$\lmi{x} \succeq 0$.  Consider the real trace
of a polynomial curve
$\mathcal{C}$, with parameterization
\begin{equation}
\label{eq:param-curve}
\begin{array}{lllll}
\Phi &:& \R &\rightarrow& \R^n \\
& & t &\mapsto &  \Phi(t) := (p_1(t), \dots, p_n(t) ) ,
\end{array} 
\end{equation}	
where $p_i(t) = \sum_{j=0}^{d_i}p_{i,j}t^j$ are univariate polynomials
in $t$ of degree $d_i$, for $i \in [m]$.  Also let
$d = \max_{i \in [m]}\{d_i\}$.
If the coefficients of the polynomials are integers, then we further
assume that the maximum coefficient's bitsize is bounded by $\tau$.

As $t$ varies over the real line, there may be several disjoint
intervals, for which the corresponding segment of $\mathcal{C}$ lies in
$\interior{S}$.  We aim to compute the endpoints, $t_{-}$
and $t_{+}$, of a maximal such interval containing $t=0$.  Let
$\bm{p}_0 = \Phi(0)$; by assumption it holds
$\bm{F}(\Phi(0)) = \bm{F}(\bm{p}_0) \succ 0$.

The input of \intersectionop (Algorithm~\ref{alg:intersectionop})
is $\bm{F}$, the LMI representation of $S$, and $\Phi(t)$, the
polynomial parameterization of $\mathcal{C}$.  Its crux is a routine,
\PEP, that solves a polynomial eigenvalue problem.  The following
lemma exploits this relation.

\begin{algorithm}
	\caption{$\intersectionop(\bm{F}, \Phi(t))$}
	\label{alg:intersectionop}
	\SetKwInOut{Input}{Input}
	\SetKwInOut{Output}{Output}
	\SetKwInOut{Require}{Require}
	
	\Input{An LMI $\lmi{x} \succeq 0$ for a spectrahedron $S$
		and a parameterization $\Phi(t)$ of a polynomial curve
		$\mathcal{C}$}
	
	\Require { $\Phi(0) \in \interior{S}$}
	
	\Output{${t_-}, {t_+}$ s.t. $\Phi({t_-}), \Phi({{t_+}}) \in \boundary{S}$}

	\BlankLine
	$T := \{ t_1 \leq t_2 \leq \cdots \leq t_{\ell} \} \leftarrow \PEP( \bm{F}(\Phi(t)) ) $\;
	${t_-} \leftarrow \max\{t \in T \,|\, t < 0\}$; // \textit{max negative polynomial eigenvalue}
	
	${t_+} \leftarrow \min\{t \in T \,|\, t >0\}$; //  \textit{min positive polynomial eigenvalue}
	
	\KwRet $t_-, {t_+}$\;
\end{algorithm}


\begin{lemma}[\PEP and  $S \cap \cC$]
	\label{lem:intersection1}
	Consider the spectrahedron
	$S = \{ \bm{x} \in \R^n \,|\, \lmi{x} \succeq 0 \}$.  Let
	$\Phi: \R \rightarrow \R^n$ be a parameterization of a polynomial
	curve $\mathcal{C}$, such that $\Phi(0) \in \interior{S}$.  Let
	$[t_{-}, t_{+}]$ be the maximal interval containing $0$, such that
	the corresponding part of $\mathcal{C}$ lies in ${S}$.  Then,
	$t_{-}$, respectively $t_{+}$, is the maximum negative, respectively  minimum
	positive, polynomial eigenvalue of $\bm{F}(\Phi(t))\bm{x}=0$, where
	{$\bm{F}(\Phi(t)) = \bm{B}_0 + t\bm{B}_1 + \cdots + t^d\bm{B}_d$}.
\end{lemma}
\begin{proof}
	The composition of $\lmi{x}$ and $\Phi(t)$ gives
	\begin{equation}
	\label{eq:FCt}
	\bm{F}(\Phi(t)) = \bm{A}_0 + p_1(t) \bm{A}_1 + \cdots + p_n(t) \bm{A}_n.
	\end{equation}
	We rewrite Expression~(\ref{eq:FCt}) by grouping the coefficients for
	each $t^i$, $i \in [d]$, then
	\begin{equation}
	\label{eq:FCBt}
	\bm{F}(\Phi(t)) = \bm{B}_0 + t\bm{B}_1 + \cdots + t^d\bm{B}_d ,
	\end{equation}
	where $\bm{B}_k = \sum_{j=0}^{n}p_{j,k}\, \bm{A}_j$, for
	$0 \leq k \leq d$.  We use the convention that $p_{j,k} = 0$, when
	$k > d_j$.
	
	For $t=0$, it holds, by assumption, that
	$\bm{F}(\Phi(0)) = \bm{B}_0 \succ 0$: point $\Phi(0)$
	lies in the interior of $S$.  Actually, for any
	$\bm{x} \in \interior{S}$ it holds $\lmi{x} \succ 0$.  On the
	other hand, if $\bm{x} \in \boundary{S}$, then
	$\lmi{x} \succeq 0$.  Our goal is to compute the maximal interval
	$[t_{-}, t_{+}]$ that contains $0$ and for every $t$ in it, we
	have $\bm{F}(\Phi(t)) \succeq 0$.

	Starting from point $\Phi(0)$, by varying $t$, we
	move on the trajectory that $\cC$ defines (in both directions)
	until we hit the boundary of $S$.  When we hit $\boundary{S}$, 
	matrix $\bm{F}(\Phi(t))$ is not strictly definite anymore.  Thus,
	its determinant vanishes.
	
	Consider the function $\theta : \R \rightarrow \R$, where
	$\theta(t) = \det \bm{F}(\Phi(t))$ is a univariate polynomial
	in $t$. If a point $\Phi(t)$ is on the boundary of the
	spectrahedron, then $\theta(t) = 0$. We opt to compute 
	$t_{-}$ and $t_{+}$, such that
	$t_{-} \leq 0 \leq t_{+}$ and
	$\theta(t_{-})= \theta(t_{+}) = 0$.
	At $t=0$, $\theta(0) > 0$ and the
	graph of $\theta$ is above the $t$-axis. So $\cC$ intersects the
	boundary when the graph of $\theta$ touches the $t$-axis for the
	first time at $t_1 \leq 0 \leq t_2$. It follows that $t_{-}=t_1$
	and $t_{+}=t_2$ are the maximum negative and minimum positive
	roots of $\theta$, or equivalently the corresponding polynomial
	eigenvalues of $\bm{F}(\Phi(t))$.
\end{proof}


\begin{lemma}\label{lem:intersection2}
  Algorithm~\ref{alg:intersectionop}, $\intersectionop(\bm{F}, \Phi(t))$,
uses $\sOO((m d)^{2.697} + m d \lg{L})$ arithmetic operations in order to
  compute the intersection, up to precision $\epsilon = 2^{-L}$, of an LMI,
  $\bm{F}$, consisting of $n$ matrices of dimension $m \times m$ with a
  parametric curve, $\Phi(t)$, of degree $d$.
\end{lemma}

\begin{proof}
  We have to construct PEP and solve it. Since $\Phi(t)$ has degree $d$, then
  {$\bm{F}(\Phi(t)) = \bm{B}_0 + t\bm{B}_1 + \cdots + t^d\bm{B}_d$} is a PEP of
  degree $d$.
  This construction costs $\OO(dnm^2)$ operations. The solving phase, using
  Lemma~\ref{lem:pep}, requires $\sOO((m d)^{2.697} + m d \lg{L})$ arithmetic
  operations and dominates the complexity bound of the algorithm.
\end{proof}

\subsection{\reflectionop}
\label{sec:reflections}

\begin{algorithm}
	\caption{\reflectionop$(\bm{F}, \Phi(t))$}
	\label{alg:reflectionop}
	\SetKwInOut{Input}{Input}
	\SetKwInOut{Output}{Output}
	\SetKwInOut{Require}{Require}

	\Input{An LMI $\lmi{x} \succeq 0$ for a spectrahedron $S$
		and a parameterization $\Phi(t)$ of a polynomial curve
		$\mathcal{C}$.}

	\Require { (i) $\Phi(0) \in \interior{S}$ \\
		(ii) $\cC$ intersects $\boundary{S}$ transversally at a smooth
		point.}

	\Output{$t_{+}$ such that $\Phi(t_{+}) \in \boundary{S}$ and the
		direction of the reflection, $\bm{s}_{+}$, at this point.  }

	\BlankLine

	$t_{-}, t_{+} \leftarrow$ \intersectionop$(\bm{F}, \Phi(t))$\;
	$\bm{w} \leftarrow \nabla\det\bm{F}(\Phi(t_{+}))$\;
	$\bm{w} \gets \tfrac{\bm{w}}{\norm{\bm{w}}}$ ; //
	\textit{Normalize $\bm{w}$}

	$\bm{s}_{+} \leftarrow \frac{d\Phi}{dt}(t_{+}) - 2 \, \langle
	\nabla\frac{d\Phi}{dt}(t_{+}),\bm{w} \rangle \, \bm{w}$\; \KwRet
	$t_{+}$, $\bm{s}_{+}$\;

\end{algorithm}

The \reflectionop operation (Algorithm~\ref{alg:reflectionop}) takes as input an LMI
 representation, $\bm{F}$,  of a spectrahedron $S$ and a polynomial curve $\cC$,
given by a parameterization $\Phi$.
Assume that $t = 0$ corresponds to a point
$\Phi(0) \in \interior{S} \cap \cC$.  Starting from $t=0$, we increase
$t$ along the positive real semi-axis.  As $t$ changes, we move along
the curve $\cC$ through $\Phi(t)$, until we hit the boundary of $S$ at
the point $\bm{p}_{+} := \Phi(t_+) \in \boundary{S}$, for some
$t_{+} > 0$.  Then, a specular reflection occurs at this point with
direction $\bm{s}_+$; this is the reflected direction.  We output
$t_{+}$ and $\bm{s}_{+}$.  Figure~\ref{fig:spectra-curve} depicts the
procedure.

The boundary of $S$, $\boundary{S}$,
with respect to the Euclidean topology,
is a subset of the real algebraic set
$\{ \bm{x} \in \R^n \,|\, \det(\lmi{x}) = 0\}$.
The latter is a real hypersurface defined by one (determinantal)
equation.
For any $\bm{x} \in \boundary{S}$ we have $\rank{\bm{F}(\bm{x})} \leq m-1$.
We assume that $\bm{p}_{+} = \Phi(t_{+})$ is such that  $\rank{\bm{F}(\bm{p}_{+})} = m-1$.
The normal direction at a point $\bm{p}\in \boundary{S}$, is the
gradient of $\det\lmi{p}$.

We compute the reflected direction using the following formula
\begin{equation}
\label{eq:reflection}
\bm{s}_{+} = \bm{u} - \tfrac{2}{\vert \bm{w}\vert^2}\langle \bm{u}, \bm{w}\rangle \,\bm{w} ,
\end{equation}
where $\bm{w}$ is the normalized gradient vector at the point
$\Phi(t_{+})$
and
$\bm{u} = \frac{d\Phi}{dt}(t_+)$ is the direction of the trajectory at
this point.  We illustrate the various vectors in
Figure~\ref{fig:spectra-curve}.

\begin{lemma}[Gradient of $\det\lmi{x}$]
	\label{lemma:gradient}
	Assume that $\bm{x} \in \boundary{S}$ and the rank of the
	$m \times m$ matrix $\lmi{x}$ is $m - 1$. Then
	\begin{equation}
	\nabla \det (\lmi{x}) =
	c \cdot (\bm{v}^\top \bm{A}_1 \bm{v}, \cdots, \bm{v}^\top \bm{A}_n \bm{v}) ,
	\end{equation}
	where $c = \frac{\mu(\lmi{x})}{\vert\bm{v}\vert^2}$, $\mu(\lmi{x})$
	is the product of the nonzero eigenvalues of $\lmi{x}$, and $\bm{v}$
	is a non-trivial vector in the kernel of $\lmi{x}$.  If
	$\rank{\bm{F}(\bm{x})} \leq m-2$, then the gradient is the zero.
\end{lemma}
\begin{proof}
	From Lemma \ref{lem:composition}:
	\begin{align}\label{eq:appendix1}
	\frac{\partial \det\lmi{x}}{\partial x_k} = \trace{\adjoint{\lmi{x}}\bm{A}_k}.
	\end{align}
	
	If $\rank{\lmi{x}}\leq - 2$, then $\adjoint{\lmi{x}}$ is the zero matrix.
	If we assume that $\rank{\lmi{x}}=m-1$, then using Lemma \ref{adjoint}:
	\begin{align*}
	  \trace{\adjoint{\lmi{x}} \bm{A}_k} & = \trace{\mu(\lmi{x}) \frac{\bm{v}\bm{u}^\top}{\bm{u}^\top \bm{v}} \bm{A}_k} \\
	  & = \frac{\mu(\lmi{x})}{\bm{u}^\top \bm{v}} \cdot \trace{\bm{v}\bm{u}^\top \bm{A}_k} = \frac{\mu(\lmi{x})}{\bm{u}^\top \bm{v}} \cdot \bm{u}^\top \bm{A}_k \bm{v}.
	\end{align*}
	However, since $\lmi{x}$ is symmetric, we can choose $\bm{v} = \bm{u}$, so:
	\begin{displaymath}
	\frac{\mu(\lmi{x})}{\bm{u}^\top \bm{v}} \cdot \bm{u}^\top \bm{A}_k \bm{v} = \frac{\mu({\lmi{x})}}{|\bm{v}|^2} \cdot \bm{v}^\top \bm{A}_k \bm{v} .
	\end{displaymath}
\end{proof}

The algorithm \reflectionop exploits Lemma~\ref{lemma:gradient}. Nevertheless,
it is not necessary to perform all then computations that the lemma indicates.
For example, because we will normalize the resulting vector and we do not need its actual
direction (internal or external), we can omit the computation of $c$. Moreover,
the nonzero vector $\bm{v}$, which satisfies $\lmi{p}\bm{v}=0$, corresponds to the
eigenvector w.r.t. the eigenvalue $t_+$ from the \PEP
(Lemma~\ref{lem:intersection1}). This is true  because
$\bm{p} = \Phi(t_+)\in \boundary{S}$ and thus $\det\bm{F}(\Phi(t_+))= 0$.

At this point we should note that we  compute the eigenvalues of \PEP up to some precision. Since
the matrix-vector multiplication is backward stable, a small perturbation
on $\bm{v}$ does not affect the computation of each coordinate of
$\nabla \det (\lmi{x})$ \cite[p.~104]{Trefethen97}.  We quantify the
accuracy of the computed $\nabla \det (\lmi{x})$ using floating point
arithmetic as follows:
\begin{lemma}\label{lem:grad_ac}
	The relative error of each coordinate of the gradient given in Lemma~\ref{lemma:gradient}
	when we compute it using floating point arithmetic with machine
	epsilon $\epsilon_{M}$ is
	$\OO(\frac{\epsilon_{M}}{\sigma_{\max}(A_i)})$, for
	$i \in [n]$,
	where $\sigma_{\max}$ is the largest singular value of $A_i$.
\end{lemma}
\begin{proof}
	Let $A\in\mathbb{R}^{m\times m}$ be a symmetric matrix and consider
	the map $f: \bm{v} \mapsto \bm{v}^TA\bm{v}$. The relative condition
	number of $f$ as defined in~\cite[p.~90]{Trefethen97} is
	
	\begin{align*}
	k(\bm{v}) = \frac{||J(\bm{v})||}{||f(\bm{v})||/||\bm{v}||} = 2\frac{||A\bm{v}||}{\bm{v}^TA_i\bm{v}} = 2\frac{\sigma_{\max}(A)}{\sigma_{\max}^2(A)} = \frac{2}{\sigma_{\max}(A)} ,
	\end{align*}
	where $J(\cdot)$ is the Jacobian of $f$. According to Theorem 15.1 in \cite[p.~111]{Trefethen97}, since the matrix-vector
	multiplication is backward stable, the relative error of each coordinate in
	the gradient computation of Lemma~\ref{lemma:gradient} is
	$\OO(\frac{\epsilon_{M}}{\sigma_{\max}(A_i)})$, for $1 \leq i \leq n$.
\end{proof}

\begin{lemma}\label{lem:reflection}
  Let $S$ be a spectrahedron represented by an LMI, $\lmi{x}$, consisting of $n$
  matrices of dimension $m \times m$. Also let $\cC$ be a parametric curve with
  parameterization $\Phi(t)$, involving polynomials of degree at most $d$.
  Algorithm~\ref{alg:reflectionop}, $\reflectionop(\bm{F}, \Phi(t))$,
  computes the intersection, up to precision $\epsilon = 2^{-L}$, of $S$ with $\cC$,
  and the reflection of $\cC$ at $\boundary{S}$, by performing
  $\sOO((m d)^{2.697}~+~m d \lg{L}~+~dnm^2)$ arithmetic operations.
\end{lemma}
\begin{proof}
  By inspecting Algorithm~\ref{alg:reflectionop} we notice that the complexity of the algorithm
  depends on the construction of $\nabla \det (\lmi{x})$
  and the call to \intersectionop.

  To compute $\nabla \det (\lmi{x})$ we just need to
  compute
  $(\bm{v}^\top \bm{A}_1 \bm{v}, \cdots, \bm{v}^\top \bm{A}_n \bm{v})$.
  If we have already computed $\bm{v}$, then this computation requires
  $\OO(nm^2)$ operations. The computation of the derivative of $\Phi(t)$ is
  straightforward, as  $\Phi$ is a univariate polynomial.
  Taking into account the complexity of \intersectionop, the total complexity for \reflectionop is $\sOO((m d)^{2.697} + m d \lg{L} + dnm^2 + nm^2)= \sOO((m d)^{2.697}~+~m d \lg{L}~+~dnm^2)$.
\end{proof}

\subsection{An example in 2D}
\begin{figure}[t]
	\begin{center}
		\centering
		\includegraphics[scale=0.19]{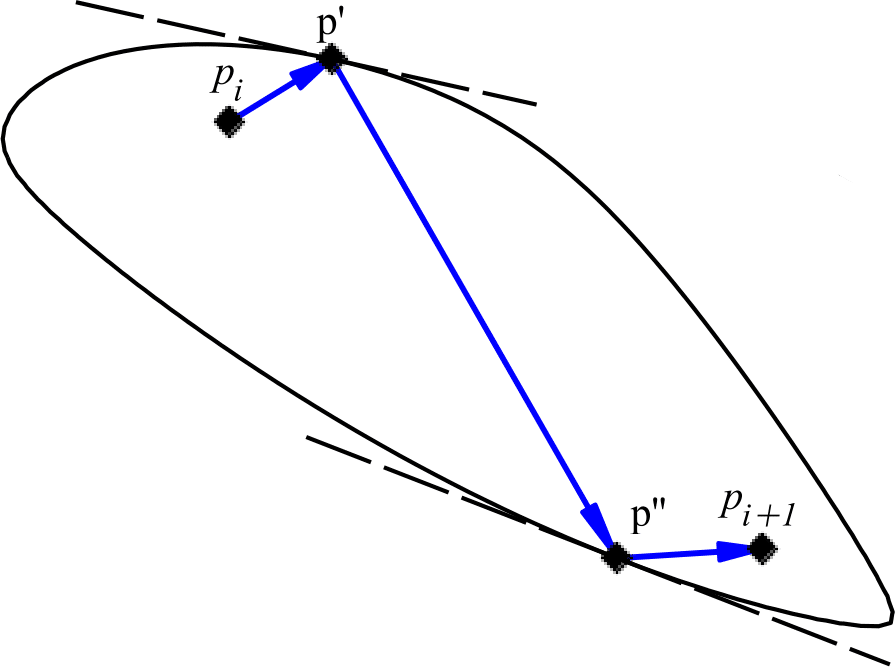}
		\includegraphics[scale=0.19]{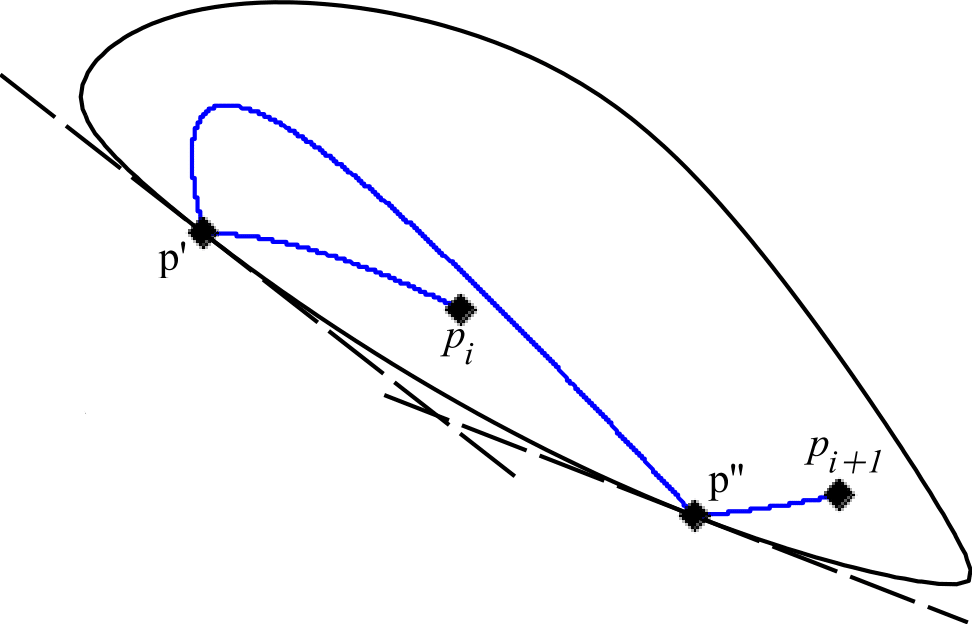}
		\caption{\emph{\textmd{The $i$-th step of the \billiard [Algorithm~\ref{alg:billiard}] (left) and of the \hmcr [Algorithm~\ref{alg:HMC}] (right) random walks.}}}
		\label{figure:examplebilliard}
	\end{center}
\end{figure}

Consider a spectrahedron $S$ in the plane (Figure~\ref{figure:examplebilliard}),
given by an  LMI $	\lmi{x} = \bm{A}_0 + x_1\bm{A}_1 + x_2\bm{A}_2$.
The  matrices $\bm{A}_i$, $0 \leq i \leq 2$, are in the appendix.

Starting from the point $\bm{p}_0 = (-1,1)^\top$, we walk along the line
$L$ with parameterization: $\Phi(t) = \bm{p}_0 + t\bm{u}$, where
$\bm{u} = (1.3,0.8)^\top$.
Then, \intersectionop finds the
intersection of $S$ with $L$, by solving the degree one PEP,
$(\bm{B}_0 + t\bm{B}_1)\bm{x}=0$, where $\bm{B}_0=\bm{F}(\bm{p}_0)$
and $\bm{B}_1 =u_1\bm{A}_1 + u_2\bm{A}_2$.
Acquiring $t_- = -0.8$ and $t_+ = 0.5$, we obtain the intersection point
$\bm{p_1}$, which corresponds to
$\bm{p}_0 + t_+ \bm{u} = (-0.3,1.4)^\top$.

To compute the direction of the trajectory, immediately after we
reflect on the boundary of $S$ at $\bm{p_1}$, \reflectionop computes
\begin{equation}
\bm{w} = \frac{\nabla\det\bm{F}(\Phi(t_+))}{\vert\nabla\det\bm{F}(\Phi(t_+))\vert} = (\bm{v}^\top\bm{A}_1\bm{v},\bm{v}^\top\bm{A}_2\bm{v} )^\top= (-0.2,-1)^\top ,
\end{equation}
where  $\bm{v}$ is the eigenvector of $(\bm{B}_0 + t\bm{B}_1)\bm{x}=0$,
with eigenvalue $t_+$. The reflected direction is
$\bm{u'} = \bm{u} - 2\langle \bm{u}, \bm{w}\rangle \,\bm{w}=(0.8,-1.3)^\top .$

\section{Random walks}
\label{sec:rand-walks}

\begin{table}[t]
	\centering
	\begin{tabular}{|c|c|}\hline
		& per-step Complexity \\ \hline\hline
		\hitandrun 	&  $\OO(m^{2.697} + m \lg{L} + nm^2)$ \\ \hline
		\cdhr 	&  $\OO(m^{2.697} + m \lg{L} + m^2)$ \\ \hline
		\billiard 	& $\sOO(\rho(m^{2.697} + m \lg{L} + nm^2))$ \\ \hline
		\hmcr 			& $\sOO(\rho((dm)^{2.697} + md \lg{L} + dnm^2))$ \\ \hline
	\end{tabular}
	\caption{\emph{\textmd{The per-step complexity of the random
				walks in
				Section~\ref{sec:rand-walks}.}}
		\label{tab:walks}}
\end{table}

Using the basic geometric operations of Section~\ref{sec:geom-ops}, we
implement and analyze three random walks for spectrahedra: Hit and Run
(\hitandrun), its variant Coordinate Directions Hit and Run (\cdhr), Billiard Walk (\billiard), and Hamiltonian Monte Carlo
with reflections (\hmcr).  In Table~\ref{tab:walks}, we present the
per-step arithmetic complexity for each random walk.

\subsection{Hit and Run}
\label{sec:hit-and-run}

\hitandrun (Algorithm~\ref{alg:hitandrun}) is a random walk that samples from any
probability distribution $\pi$ truncated to a convex body $K$; in our case a
spectrahedron $S$.
We should mention that there exist bounds for its mixing time only when
$\pi$ is log-concave distribution, for example the uniform distribution, which
is $\sOO(n^3)$.

At the $i$-th step, \hitandrun chooses uniformly at random a (direction of a)
line $\ell$, passing from its current position $\bm{p}_i$. Let
$\bm{p}_1$ and $\bm{p}_2$ be the intersection points of $\ell$ with $S$.
Let $\pi_\ell$ be the restriction of $\pi$ on the segment $\segment{p_2}{p_2}$.
Then, we choose $\bm{p}_{i+1}$ from $\segment{p_1}{p_2}$ w.r.t. the distribution
$\pi_\ell$.

\begin{algorithm}
	\caption{\textsc{Hit-and-Run\_Walk} \ \ \ (\hitandrun)}
	\label{alg:hitandrun}
	\SetKwRepeat{Do}{do}{while}
	\SetKwInOut{Input}{Input}
	\SetKwInOut{Output}{Output}
	\SetKwInOut{Require}{Require}
	
	\Input{LMI $\lmi{x} \succeq 0$ for a spectrahedron $S$
		\&~a~point~$\bm{p}_i$.}
	
	\Require { $\bm{p}_i \in S$}
	
	\Output{The point $\bm{p}_{i+1}$ of the $(i+1)$-th step of the walk. }

	{BO ($\bm{F}$, interior point $\bm{p_i}$)}{
	$\bm{v} \pickrandom \mathcal{U}(\boundary{\unitball})$;  // \textit{choose direction} \\
	$\Phi(t) := \bm{p}_i + t\bm{v}$; // \textit{define trajectory}\\
	
	${t_-}, {t_+}$ $\leftarrow$ \intersectionop($\bm{F}, \Phi(t)$)\;
	$\bm{p}_{i+1} \pickrandom [\bm{p}_i+{t_-}\bm{v},\ \bm{p}_i+{t_+}\bm{v}]$ w.r.t. $\pi_{\ell}$\;
	
	\KwRet $\bm{p}_{i+1}$\;	
	}
\end{algorithm}





\begin{lemma}
  \label{lem:HnR}
  The per-step complexity of \hitandrun is $\sOO(m^{2.697} + m \lg{1/\epsilon} + n m^{2})$.
\end{lemma}
\begin{proof}
  The per-step complexity of \hitandrun is dominated by the
  \intersectionop, which requires $\OO(nm^2)$ operations for the
  construction of the \PEP and $\sOO(m^{2.697} + m \lg{1/\epsilon})$ for
  solving it; in the case where we want to approximate the intersection point up to
  a factor or $\epsilon = 2^{-L}$ (Lemma~\ref{lem:intersection2}).
\end{proof}

There is also a variation of \hitandrun, the \emph{coordinate directions Hit and
  Run} (\cdhr )~\cite{Smith84}. This walk chooses the direction vector randomly
and uniformly among the vector basis \{$\bm{e}_i,\ i \in [n]\}$.
In \cdhr, for every step aside the first, the construction of the \PEP takes
$\OO(m^2)$ operations, and the complexity does not depend on the dimension $n$.
The reason for this improvement is that to build the \PEP,
$\bm{F}(\bm{p}_i + t\bm{e}_j) = \bm{F}(\bm{p}_i) + t\bm{A}_j$,
we can obtain the  value of
$\bm{F}(\bm{p}_i)$ from
$\bm{F}(\bm{p}_i) = \bm{F}(\bm{p}_{i-1}) + \hat{t}\bm{A}_k$, assuming that at the
previous step we have chosed $\bm{e}_k$ as direction. There is no theoretical
mixing time for \cdhr.



\subsection{Billiard walk}
\billiard \cite{POLYAK20146123}, Algorithm~\ref{alg:billiard}, samples a convex body
$K$ under the uniform distribution; no theoretical results for its mixing time
exist. At $i$-th step, being at position $\bm{p}_i$, it chooses uniformly a
direction vector $\bm{v}$ and a number $\ell$, where $\ell=-\tau\ln\eta$ and
$\eta\sim U(0,1)$. Then, it moves at the direction of $\bm{v}$ for distance $\ell$.
If during the movement, it hits the boundary without having covered the required
distance $\ell$, then it continues on a reflected trajectory. If the number of
reflections exceeds a bound $\rho$, it stays at $\bm{p}_i$. In
\cite{POLYAK20146123} they experimentally conclude that \billiard mixes faster
when $\tau \approx diam(K)$, where $diam(K)$ is the diameter of $K$.

\begin{algorithm}[h]
	\caption{\textsc{Billiard\_Walk} \ \ \ (\billiard)}
	\label{alg:billiard}
	\SetKwRepeat{Do}{do}{while}
	\SetKwInOut{Input}{Input}
	\SetKwInOut{Output}{Output}
	\SetKwInOut{Require}{Require}
	
	\Input{An LMI $\lmi{x} \succeq 0$ for a spectrahedron $S$, a point
		$\bm{p_i}$, the diameter $\tau$ of $S$ and a bound $\rho$ on the
		number of reflections.}
	
	\Require { $\bm{p}_i \in S$}
	
	\Output{The point $\bm{p}_{i+1}$ of the $(i+1)$-th step of the walk. }

	$\ell \leftarrow -\tau\ln\eta$ ; $\eta \pickrandom \mathcal{U}((0,1))$; // \textit{choose length}\\
	$\bm{v} \pickrandom \mathcal{U}(\boundary{\unitball})$;  // \textit{choose direction} \\
	$\bm{p} \leftarrow \bm{p}_i$\;
	
	\Do{$\ell > 0$}{
		$\Phi(t) := \bm{p} + t\bm{v}$; // \textit{define trajectory}\\
		$t_+, \bm{s}_+$ $\leftarrow$ \reflectionop($\bm{F}, \Phi(t)$)\;
		$\hat{t} \leftarrow \min\{t_+, \ell\}$ ;
		$\bm{p} \leftarrow \Phi(\hat{t}) $ ;

		\lIf{$\overline{t} < \ell$} {
		$\bm{v} \leftarrow$ $\bm{s}_+$
		}
		$\ell \leftarrow \ell - \hat{t}$ ;
	}
	
	\lIf {$\#\{\text{reflections}\} > \rho$}{
		\KwRet $\bm{p_{i+1}}= \bm{p_{i}}$
	}
	{
		\KwRet $\bm{p}_{i+1} = \bm{p}$	
	}

\end{algorithm}

\begin{lemma}
  \label{lem:billard}
  The per-step complexity of \billiard is  $\sOO(\rho(m^{2.697} + m \lg{L} + nm^2))$,
  where $\rho$ is the number of reflections.
\end{lemma}
\begin{proof}
  The per-step complexity of \billiard is dominated by the \reflectionop, which
  requires $\sOO(m^{2.697} + m \lg{L} + nm^2)$ arithmetic operations
  (Lemma~\ref{lem:reflection}), when we want to approximate the intersection
  point up to a factor of $2^{-L}$. Since we allow at most $\rho$ reflections
  per step, the total complexity becomes
  $\sOO(\rho(m^{2.697} + m \lg{L} + nm^2))$.
\end{proof}

\subsection{Hamiltonian Monte Carlo with Reflections}
\label{sec:hmcr}

\begin{algorithm}
  \caption{\textsc{HMC\_w\_reflection} \ \ \ (\hmcr)}
  \label{alg:HMC}
  \SetKwRepeat{Do}{do}{while}
  \SetKwInOut{Input}{Input}
  \SetKwInOut{Output}{Output}
  \SetKwInOut{Require}{Require}

  \Input{An LMI $\lmi{x} \succeq 0$ representing a spectrahedron $S$,
    a point $\bm{p}_i$, the diameter $\tau$ of $S$ and a bound $\rho$ to the number of reflections.}

  \Require { $\bm{p}_i \in S$}

  \Output{The point $\bm{p}_{i+1}$ of the $(i+1)$-th step of the walk. }

  $\ell \leftarrow \tau\eta$;  $\ \eta \pickrandom \mathcal{U}((0,1))$; // \textit{choose length}\\
  $\bm{v} \pickrandom \mathcal{N}(0,I_n)$; // \textit{choose direction}\\

  \Do{$\ell > 0$}{
    Compute trajectory $\Phi(t)$ from ODE (\ref{eq:ode_ham})\;
    ${t_+}, \bm{s}_+$ $\leftarrow$ \reflectionop($\bm{F}, \Phi(t)$)\;
    $\hat{t} \leftarrow \min\{{t_+}, \ell\}$ ;
    $\bm{p} \leftarrow \Phi(\hat{t})$ ;
    $\bm{v} \leftarrow$ $\bm{{s}}$ ;
    $\ell \leftarrow \ell - \hat{t}$ ;
  }

  \lIf {$\#\ \{\text{reflections} \}> \rho$}{
    \KwRet $\bm{p}_{i+1}= \bm{p}_{i}$
  }

  \KwRet $\bm{p}_{i+1} = \bm{p}$	\;

\end{algorithm}

Hamiltonian Monte Carlo (HMC), can be used to sample from any probability
distribution $\pi$. Our focus lies again on the log-concave distributions, that
is $\pi(\bm{x})\propto e^{-\alpha f(\bm{x})}$. We exploit the setting in~\cite{Lee18} as they approximate the Hamiltonian trajectory with a polynomial curve. In this setting, if we assume
that $f$ is a strongly convex function, then the mixing time of HMC is
$\OO(k^{1.5}\log(n/\epsilon))$, where $\kappa$ is the condition number of
$\nabla^2f$~\cite{Lee18}. If we truncate $\pi$ by considering its restriction in
a convex body, then we can use boundary reflections (\hmcr), as in Algorithm~\ref{alg:HMC}, to ensure that the
random walk converges to the target distribution \cite{Chevallier18}; however,
in this case the mixing time is unclear.

The Hamiltonian dynamics behind HMC operate on a $n$-dimensional
position vector $\bm{p}$ and a $n$-dimensional momentum
$\bm{v}$. So the full state space has $2n$ dimensions. The
system is described by a function of $\bm{p}$ and $\bm{v}$ known as
the Hamiltonian, 
\[ H(\bm{p}, \bm{v}) = U(\bm{p}) + K(\bm{v})
  = f(\bm{p}) + \frac{1}{2}\vert\bm{v}\vert^2 .
\]
To sample from $\pi$, one has to solve the
following system of Ordinary Differential Equations (ODE):
{
	\begin{equation}\label{eq:ode_ham}
	\begin{split}
	\frac{d\bm{p}}{dt} = \frac{\partial H(\bm{p},\bm{v})}{\partial \bm{v}}\\
	\frac{d\bm{v}}{dt} = -\frac{\partial H(\bm{p},\bm{v})}{\partial \bm{p}} 
	\end{split}
	\quad\Rightarrow\quad
	\left\{
	\begin{array}{lll}	
	\frac{d\bm{p}(t)}{dt} = \bm{v}(t)\\
	\quad\\
	\frac{d\bm{v}(t)}{dt} = -\alpha\nabla f(\bm{p})\\
	\end{array} 
	\right. .
	\end{equation}
}
If $\pi(\bm{x})$ is a log-concave density, then we can approximate the
solution of the ODE with a low degree polynomial trajectory
\cite{Lee18}, using the collocation method. A degree
$d=\OO(1/\log(\epsilon) )$ suffices to obtain a polynomial trajectory
with error $\OO(\epsilon)$, for a fixed time interval, while we 
perform just $\sOO(1)$ evaluations of $\nabla f(\bm{x})$.

HMC at the $i$-th step uniformly samples a step $\ell$ from a proper
interval to move on the trajectory implied by ODE (\ref{eq:ode_ham}),
choses $\bm{v}$ randomly from $\mathcal{N}(0, I_n)$, and updates
$\bm{p}$ using the ODE in (\ref{eq:ode_ham}), for $t\in[0, \ell]$.
When $\pi$ is truncated in a convex body, then \hmcr fixes an upper bound $\rho$
on the number of reflections and reflects a polynomial trajectory as
we describe in Section~\ref{sec:reflections}.

\begin{figure}[t]
	\centering
	\includegraphics[width=0.48\textwidth]{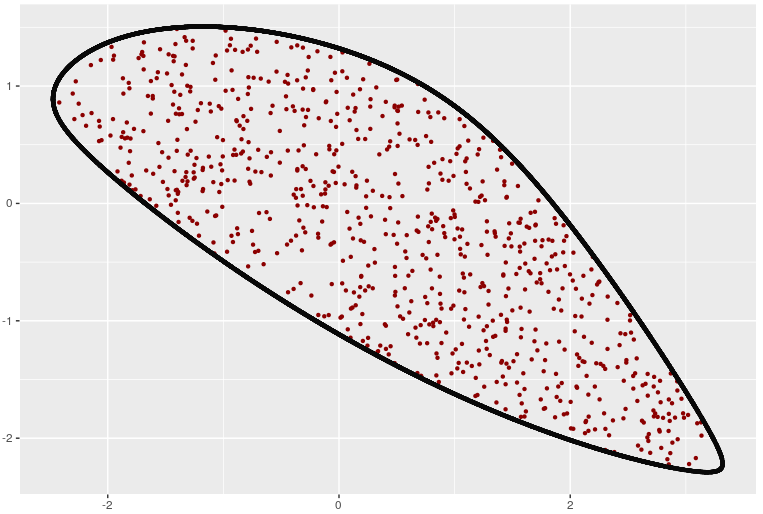}
	\hspace{0.2cm}
	\includegraphics[width=0.48\textwidth]{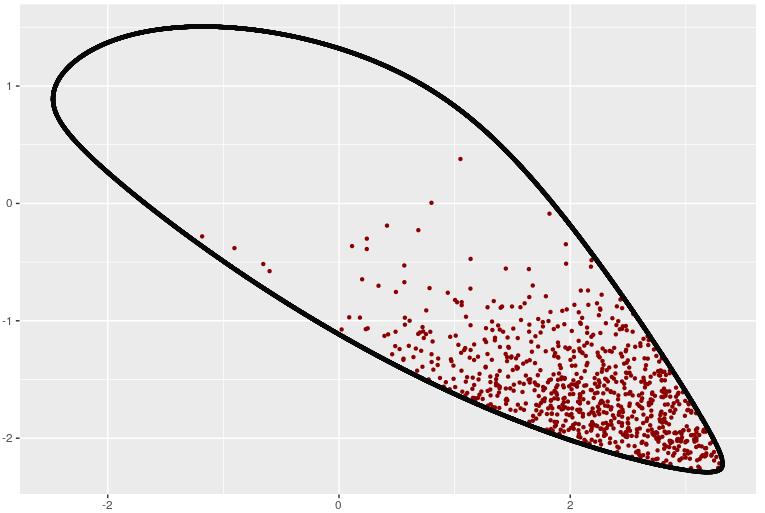}
	\caption{\emph{\textmd{Samples from the uniform distribution with \billiard\ (left) and from the Boltzmann distribution $\pi(x)\propto e^{-cx/T}$, where $T=1,\ c=[-0.09, 1]^T$, with \hmcr\ (right). The volume of this spectrahedron is $10.23$.}}\label{fig:sampling}}
\end{figure}

Each step of \hmcr, when $\pi(x)$ is a log-concave density truncated
by $S$, costs $\sOO(\rho((dm)^{2.697} + md \lg{L} + dnm^2))$,
if we approximate the intersection points up to a factor $2^{-L}$,
where $d$ is the degree of the polynomial that approximates the
solution of the~ODE~(\ref{eq:ode_ham}).

\begin{lemma}
  \label{lem:hmc}
  The per-step complexity of \hmcr  is
  $\sOO(\rho((dm)^{2.697} + md \lg{L} + dnm^2))$,
  where $\rho$ is the number of reflections.
\end{lemma}

\section{Applications and experiments}
\label{sec:implementation}

This section demonstrates and compares the algorithms of
Section~\ref{sec:rand-walks} and the efficiency of our software on three
applications that rely on sampling from spectrahedra. 

We call \emph{walk length} the number of the intermediate
points that a random walk visits before producing a single sample. 
The longer the walk length of a random walk is, the smaller the distance of the current distribution to the stationary (target) distribution becomes. 
Typically we choose a sufficiently large length for the first sample, this procedure is often called "burning".

Our code is parameterized by the
floating point precision of the computations. 
We use {\tt Eigen}~\cite{eigenlib} for basic linear algebra operations, such as
Cholesky decomposition and matrix multiplication. For eigenvalue computations,
we employ {\tt Spectra}~\cite{Spectra}, which is based on {\tt Eigen} and offers
crucial optimizations. First, it solves generalized eigenvalue problems
 of special structure;
that is $(B_0 - \lambda B_1)v =0$, when $B_0$ is positive
semidefinite and $B_1$ symmetric. This operation is encountered when \billiard\ or
\hitandrun\ call \intersectionop. Second, it offers directly the computation
of the largest eigenvalue which corresponds to $t_+$ after a simple
transformation. Finally, {\tt Spectra} provides $\sim 20$x speedup over
the default eigenvalue computation by {\tt Eigen}.
To the best of our knoledge, our software is the first that can sample
efficiently from spectrahedra and estimates volumes up to a few hundred
dimensions. It is accessible on github.%
\footnote{\scriptsize\url{https://github.com/GeomScale/volume_approximation/tree/sample_spectrahedra}}

For our experiments, we generate random spectrahedra following
\cite{Dabbene2010ARC}. In particular, to construct the LMI of Equation~(\ref{eq:lmi}) we set $A_0$ to be positive semidefinite, i.e., $A_0 = ZZ^T + I_m$, where we pick the elements of $Z\in\mathbb{R}^{m\times m}$ uniformly at random from $[0,1]$. Then, for $A_i,\ i=1,\dots ,n$ we set,
\[
A_i =
\left[
\begin{array}{cc}
\tilde{Q} & 0 \\
0 & -\tilde{Q}
\end{array}
\right] ,\ \tilde{Q} = Q + Q^T ,
\]
where we pick the elements of $Q\in\mathbb{R}^{(m/2) \times (m/2)}$ uniformly at random from $[-1,1]$.
We performed all the experiments on PC with
{\tt Intel Core i7-6700 3.40GHz $\times$ 8 CPU} and {\tt 32GB RAM}.

\subsection{Volume computation}
\label{sec:volume}

We use the geometric operations (Section~\ref{sec:geom-ops})
and the random walks (Section~\ref{sec:rand-walks})
to compute the volume $\vol(S)$ of spectrahedron $S$.
Our implementation approximates $\vol(S)$ within relative error
$0.1$ with high probability in a few minutes, for dimension $n=100$.

A typical randomized algorithm for volume approximation exploits a
Multiphase Monte Carlo (MMC) technique, which reduces volume
approximation of convex body $S$ to computing a telescoping product of
ratios of integrals over $S$. In particular, for any sequence of functions $\{f_0,\ \dots ,f_k\}$, where $f_i:\R^n\rightarrow \R$,
we have:

	\begin{equation}
	\label{telegeneric}
	\vol(S) = \int_S 1dx = \int_S f_k(x)dx\frac{\int_S f_{k-1}(x)dx}{\int_S f_k(x)dx}\cdots\frac{\int_S 1dx}{\int_S f_0(x)dx} .
	\end{equation}
Notice that $\frac{\int_P f_{i-1}(x)dx}{\int_P f_i(x)dx} = \int_P\frac{f_{i-1}(x)}{f_i(x)}\frac{f_i(x)}{\int_P f_i(x)dx}dx$.
To estimate each ratio of integrals, we sample $N$ points from a
distribution proportional to $f_i$ and, we use the unbiased estimator
$\frac{1}{N}\sum_{j=1}^N\frac{f_{i-1}(x_j)}{f_i(x_j)}$. To exploit
Equation~(\ref{telegeneric}) we have to (i) fix the sequence such that $k$
is as small as possible, (ii) select $f_i$'s such
that we can compute efficiently each integral ratio, and (iii) compute
$\int_P f_k(x)dx$. The best theoretical result of \cite{Cousins15} fixes a sequence of
spherical Gaussians $\{f_0,\ \dots ,f_k\}$ with the mode being in $S$,
parameterized by the variance.
The overall complexity is $\sOO(n^3)$ \membershipop calls. The implementation in
\cite{Cousins16} is based on this algorithm but handles only convex
polytopes in H-representation as it requires the facets of the polytope and an inscribed
ball to fix the sequence of
Gaussians. Both the radius of the inscribed ball and the number of facets
strongly influence the performance of the algorithm. So, it cannot
handle efficiently the case of convex bodies without a facet
description, e.g., zonotopes \cite{Cousins16}, as it results  a
big sequence of ratios that spoil practical~efficiency.

\begin{table}[t]
	\centering
	\begin{tabular}{|c|ccc|}\hline
		$S$-$n$-$m$ & $\mu\pm t_{\alpha,\nu-1}\frac{s}{\sqrt{\nu}}$ &  $Points$ & $Time\ (sec)$ \\ \hline\hline
		$S$-$40$-$40$ & $(1.34\pm 0.12)$e-06 & 9975.2 & 6.7 \\ \hline
		$S$-$60$-$60$ & $(1.23\pm 0.11)$e-20 &  20370.9 & 28.5 \\ \hline
		$S$-$80$-$80$ & $(4.24\pm 0.26)$e-33 &  31539.1 & 124.4 \\ \hline
		$S$-$100$-$100$ & $(1.21\pm 0.10)$e-51  &  52962.7 & 362.3 \\ \hline
	\end{tabular}
	\caption{\emph{\textmd{(Volume of spectrahedra) For each $S$-$n$-$m$ we run $\nu=10$ experiments; $m$ is the matrix dimension in LMI and $n$ the ambient dimension. 
$\mu$ stands for the average volume, $s$ for the standard deviation; we give a confidence interval with level of confidence $\alpha = 0.05$; $t_{\alpha,\nu-1}$ is the critical value of student's distribution with $\nu-1$ degrees of freedom. $Points$ denotes the average number of points generated and $Time$ the
				average runtime in seconds.
For all the above we set the error parameter $e=0.1$.}} \label{tab:volumes}}
\end{table}

Our approach is to consider the $f_i$'s as a sequence of indicator functions of concentric balls centered in $S$, as in \cite{Emiris18}. In particular, let $f_k$ and $f_0$ be the indicator functions of $r\unitball$ and $R\unitball$ respectively, while $r\unitball\subseteq S \subseteq R\unitball$ and $S_i=(2^{(k-i)/n}r\unitball)\cap S$ for $i=0,\dots ,k$.
Thus, it suffices to
compute $\vol(r\unitball)$ and apply the following:
\begin{equation}
\label{eq:mmc_vol}
\vol(S) = \vol(S_k) \frac{\vol(S_{k-1})}{\vol(S_k)}
\cdots
\frac{\vol(S_0)}{\vol(S_{1})} ,
\quad  k=\lceil n\lg (R/r)\rceil .
\end{equation}
Furthermore, we employ the annealing schedule from \cite{Chalkis19}
to minimize $k$, without computing neither an
enclosed ball $r\unitball$ nor an enclosing ball $R\unitball$ of $S$.
We do so by probabilistically bounding each ratio of Equation~(\ref{eq:mmc_vol}) in an interval $[r, r+\delta]$, which is  given as
input.
To approximate each ratio of volumes, we sample uniformly distributed points from $S_i$ and count points in $S_{i-1}$. We follow the experimental results of Section~\ref{sec:control} and use \billiard which mixes faster than \hitandrun. 

Table~\ref{tab:volumes} reports the average volume, runtime, number of
points generated for each $S$-$n$-$m$ over $10$ trials. We also compute a $95\%$ confidence interval for the volume. Notice that for all cases the extreme values of each interval imply an error $\leq 0.1$, which was the requested error.
For $n=40$ just a few seconds suffice to approximate the volume and for $n=100$ our implementation takes a few minutes. 

\subsection{Expected value of a function}
\label{sec:control}

Randomized algorithms are commonly used for problems in robust control
analysis to overcome the (worst case) hardness, especially in
probabilistic
robustness~\cite{Blondel00, Khargonekar96,TCD-rc-bk-12,CalCam-rc-cdc-03}. A central problem is to approximate
the integral of a function over a spectrahedron,
e.g.~\cite{Calafiore-cdc-04, Ray93} and thus uniform sampling is of
particular interest. To put our experiments into perspective, we
present experiment up to $n=200$, while in \cite{Calafiore-cdc-04} and
\cite{CalCam-rc-cdc-03} they use only \hitandrun\ for experiments in
$n\leq 10$.

Our goal is to compute the expected value
of a function $f:\R^n\rightarrow [0,1]$, with respect to the measure
given by the uniform distribution $\pi$ over $S$,
i.e.,~$I = \int_Sf(x)\pi(x)dx$.
A standard approach
is the Monte Carlo method, which suggests to
sample $N$ independent samples from $\pi$. Then, 
\[\hat{\mathbb{E}}_N[f] = \frac{1}{N}\sum_{i=1}^Nf(x_i)\] is
an unbiased estimator for $I$.
We employ the random walks of Section~\ref{sec:rand-walks} 
to sample uniformly distributed points from $S$ (i.e.,~\hitandrun, \cdhr, and \billiard)
and we experimentally compare their efficiency.
It turns out that \billiard\ mixes much faster and results to better accuracy (see Figures~\ref{fig:control1} \&~\ref{fig:control2}). This observation agrees with the experiments on the rate of convergence for \hitandrun\ and \billiard\ in \cite{POLYAK20146123}. To come to a decisive conclusion we need to perform a more detailed practical study on the mixing time of these random walks; we leave this study as future work.
\begin{figure}[t]
	\centering
	\includegraphics[width=0.9\textwidth]{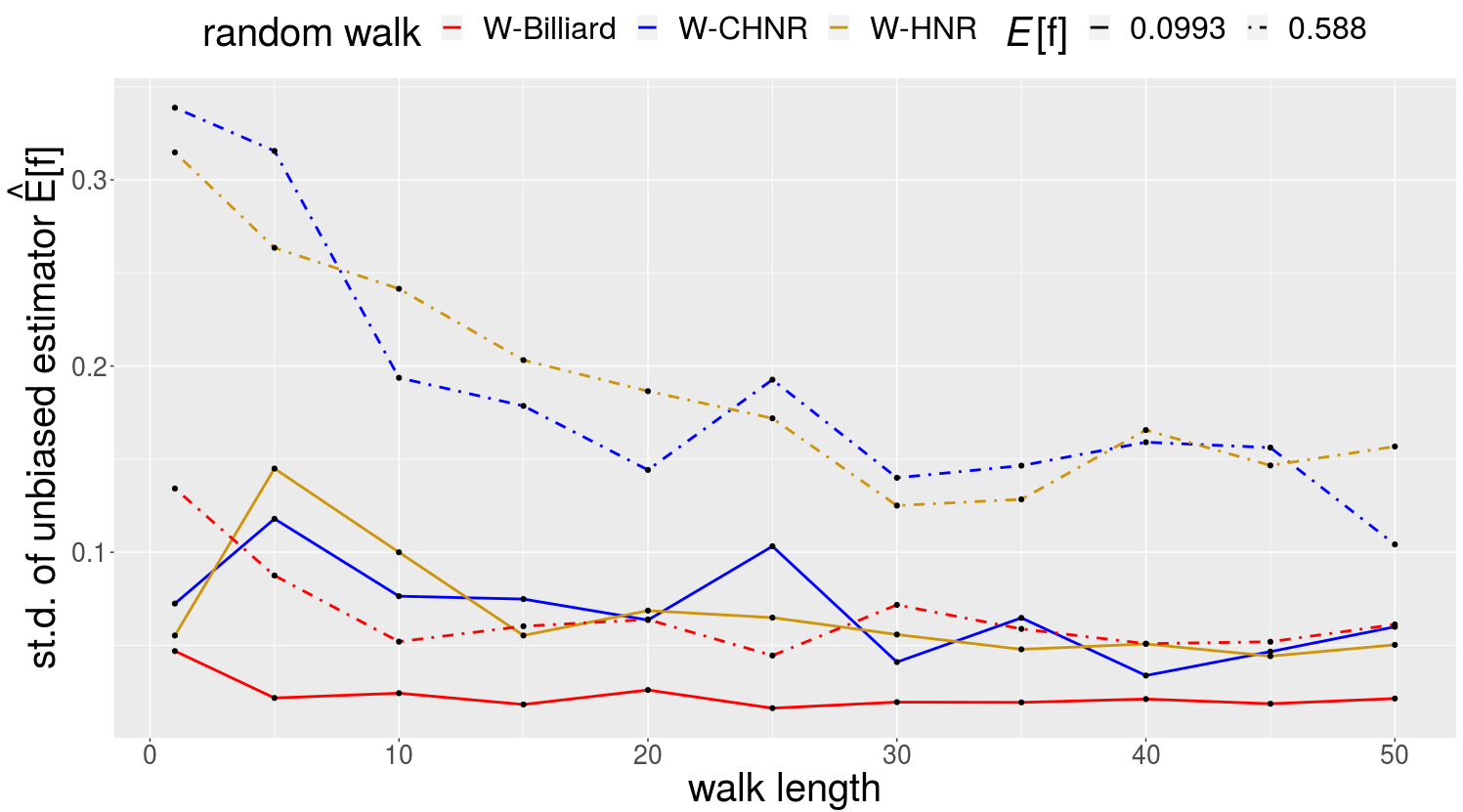}
	\caption{\emph{\textmd{The standard deviation of  $\hat{\mathbb{E}}_N[f]$ over $M=20$ trials, estimating 2 functions with $\mathbb{E}[f_1]=0.0993$ and $\mathbb{E}[f_2]=0.588$. For each walk length we sample $N=200$ points and we~repeat~$M=20$~times.}}\label{fig:control1}}
\end{figure}

The variance of an estimator is a crucial as it bounds the
approximation error. Using Chebyshev's inequality and \cite{Lovasz97},
we have
\begin{equation}\label{eq:estimator_error}
Prob[|\hat{\mathbb{E}}_N[f] - \mathbb{E}[f]|\leq \epsilon] \leq \frac{var(\hat{\mathbb{E}}_N[f])}{\epsilon^2}\leq\frac{4M_{\epsilon}}{N\epsilon^2},
\end{equation}
where $M_{\epsilon}$ is the mixing time of the random walk one uses to
sample ``$\epsilon$ close" to the uniform distribution from $S$. Thus,
for fixed $N$ and $\epsilon$, the smaller the mixing time of the
random walk is, the smaller the variance of estimator $\hat{\mathbb{E}}_N[f]$ and hence, the better the approximation.
We estimate $I$ when $f:=1(S\cap H)$, where $H$ is the union of two
half-spaces $H:=\{ x\ |\ cx\leq b_1\ \text{ or }\ cx\geq b_2\}$, where $b_2>b_1$ and
$1(\cdot )$ is the indicator function. Note that
$I = \vol(S\cap H)/\vol(S)$. To estimate it we sample approximate uniformly distributed points from $S$ and we count the number of points that lie in $S\cap H$.


\begin{figure}[t]
	\centering
	\includegraphics[width=0.9\textwidth]{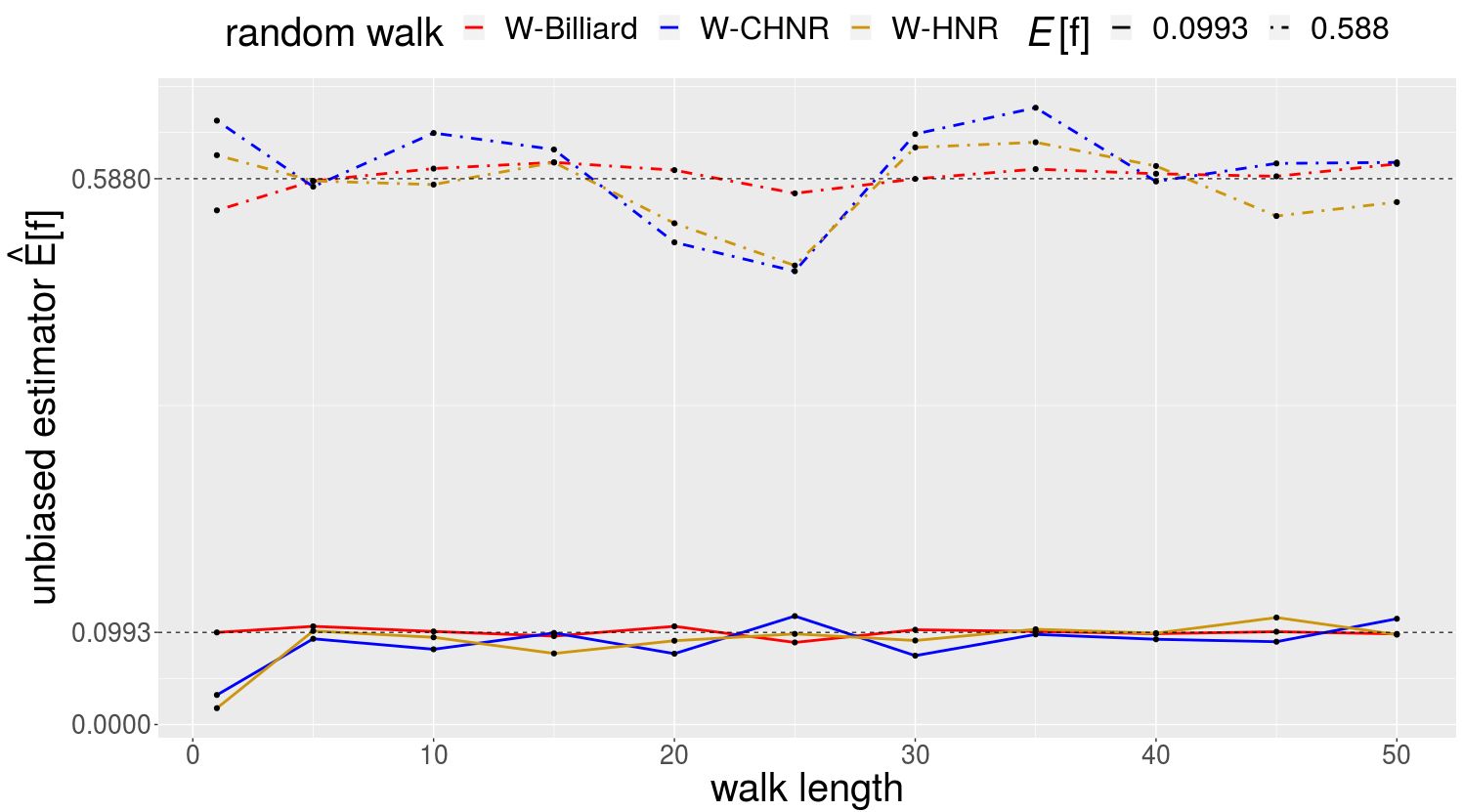}
	\caption{\emph{\textmd{The mean value of the estimator $\hat{\mathbb{E}}_N[f]$ over
				$M=20$ trials, estimating two functions with
				$\mathbb{E}[f_1]=0.0993$ and $\mathbb{E}[f_2]=0.588$. For each walk
				length we sample $N=200$ points and we repeat $M=20$
				times.}}\label{fig:control2}}
\end{figure}

\begin{table}[h]
	\centering
	\begin{tabular}{|c|ccccccc|}\hline
		walk length & $1$ & $5$ & $10$ & $20$ & $30$ & $40$ & $50$ \\ \hline\hline
		$S$-$100$-$100$ & 1.4 & 3.2 &  7.7  & 9.5  & 16.1 & 21.4 & 28.2 \\ \hline
		$S$-$200$-$200$ & 16.2 & 75.7 &  148 & 303 & 443 & 584 & 722 \\ \hline
	\end{tabular}
	\caption{\emph{\textmd{Average time in sec to sample $200$ points with \billiard\ from $10$ random spectrahedra $S$-$n$-$m$; $n$ for the dimension that $S$-$n$-$m$ lies; $m$ for the dimension of the matrix in LMI.}}\label{tab:bill_times}}
\end{table}

We estimate two functions $f_1,\ f_2$ with $\mathbb{E}[f_1]=0.0993$ and
$\mathbb{E}[f_2]=0.5880$ in dimension $n=50$ and for various walk lengths. For
each walk length we sample $N=200$ points and we repeat $M=20$
times. Then, for each $N$-set we compute
$\frac{1}{N}\sum_{i=1}^Nf(x_i)$ and we take the average and the
standard deviation (st.d.) over $M$. Figures~\ref{fig:control1}, \&~\ref{fig:control2} illustrates these values,
while the walk length increases. Notice that the st.d.\ is much
smaller and the approximation more stable when \billiard is used
compared to both \hitandrun and \cdhr.  As \billiard mixes faster, we
report in Table~\ref{tab:bill_times} the average time our software
needs to sample $N=200$ points for various walk lengths for
\billiard\ in $n=100,\ 200$. The average time to generate a point is
$\approx 0.3$ and $\approx 7.2$ milliseconds respectively.

\subsection{Sampling from non-uniform distributions}
\label{sec:optimization}

The random walks of Section~\ref{sec:rand-walks} open a promising
avenue for approximating the optimal solution of a semidefinite
program, that is
\begin{equation}\label{eq:sdp}
\min \langle \bm{c}, \bm{x}\rangle ,\ \text{subject to }\bm{x}\in S.
\end{equation}
We parameterize the optimization algorithm in~\cite{Kalai06} with the
choice of random walk and demonstrate that its efficiency relies
heavily on the sampling method. We perform experiments with \hmcr and
\hitandrun, as both can sample from the distribution the algorithm
requires.  Deterministic approximations to the optimal solutions of
these tests, were acquired via the {\tt SDPA} library~\cite{sdpa}.

\begin{figure}[t]
	\centering
	\includegraphics[width=\textwidth]{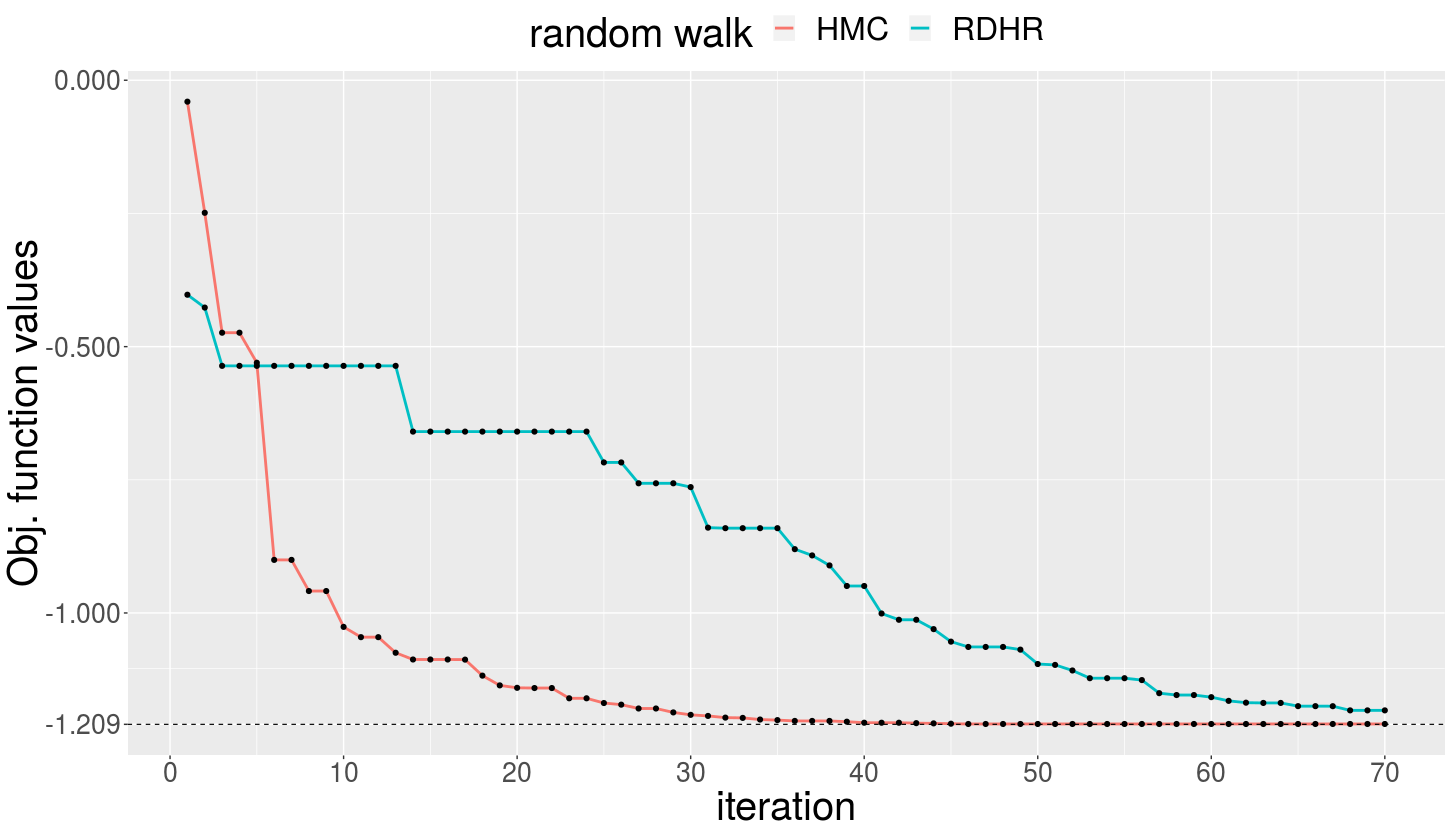}
	\caption{\emph{\textmd{Sample a point from $\pi(x)\propto e^{-cx/T_i}$ and update the objective current best in each iteration, with $T_0 \approx diam(S)$ and $T_i = T_{i-1}(1-1/\sqrt{n}),\ i=1,\dots 70$. The walk length equals to one for \hmcr and $500+4n^2=10\,500$ for \hitandrun.}}\label{fig:hmc_hnr}}
\end{figure}

The strategy to approximate the optimal solution $\bm{x}^*$ of
Equation~(\ref{eq:sdp}), is based on sampling from the Boltzmann
distribution, i.e.,~$\pi(\bm{x})\propto e^{-\bm{c} \bm{x}/T}$,
truncated to $S$.  The scalar $T$, is called \emph{temperature}.  As
the temperature $T$ diminishes, the mass of $\pi$ tends to concentrate
around its mode, which is $\bm{x}^*$. Thus, one could obtain a uniform
point using the algorithm in \cite{Lovasz03}, and then use it as a
starting point to sample from $\pi_0\propto e^{-cx/T_0}$, where
$T_0=R$ and $S\subseteq R\unitball$.  Then, the cooling schedule
$T_{i+1} = T_i(1-1/\sqrt{n})$ guarantees that a sample from $\pi_i$
yields a good starting point for $\pi_{i+1}$. After $\sOO(\sqrt{n})$
steps the temperature will be low enough, to sample a point within
distance $\epsilon$ from $\bm{x}^*$ with high probability.

\noindent
In  \cite{Kalai06}, they use only \hitandrun. We also employ  \hmcr.
To sample from Boltzmann distributions with \hmcr, at each step, starting from  $\bm{p_i}$ and with momenta $\bm{v_i}$,
the ODE of Equation~(\ref{eq:ode_ham}) becomes
\begin{equation}\label{eq:boltz_ode}
\frac{d^2}{dt^2}\bm{p}(t) = -\frac{\bm{c}}{T},\ \frac{d}{dt}\bm{p}(0) = \bm{v_i},\ \bm{p}(0) = \bm{p_i} .
\end{equation}
Its solution is the polynomial
$\bm{p}(t) = -\frac{\bm{c}}{2T}t^2 + \bm{v_i}t + \bm{p_i}$, which is a
parametric representation of a polynomial curve, see
Equation~(\ref{eq:param-curve}).

\begin{table}[t]
	\centering
	\begin{tabular}{|c|ccc|}\hline
		$S$-$n$-$m$ & \hmcr & \hitandrun\ $W_1$ & \hitandrun\ $W_2$ \\ \hline\hline
		$S$-$30$-$30$ &  $20.1$ / $2.9$/ $0$  & $184.3$ / $3.4$ / $1$  & $52.1$ / $5.2$ / $0$  \\ \hline
		$S$-$40$-$40$ &  $24.6$ / $7.9$ / $0$  & $223.3$ / $9.9$ / $2$  &  $61.9$ / $17.1$ / $0$    \\ \hline
		$S$-$50$-$50$ &  $29.2$ / $12.7$ / $0$  & $251.2$ / $22.3$ / $3$  &  $72.3$ / $44.6$ / $0$   \\ \hline
		$S$-$60$-$60$ &  $32.8$ / $24.32$ / $0$  & $272.7$ / $41.1$ / $3$  &  $81.5$ / $98.9$ / $0$  \\ \hline
	\end{tabular}
	\caption{\emph{\textmd{The average $\#$iteration / runtime / failures over $10$ generated $S$-$n$-$m$,  to achieve relative error $\epsilon\leq 0.05$. The walk length is one for \hmcr\ and $W_1=4\sqrt{n}$ and $W_2=4n$ for \hitandrun. With "failures" we count the number of times the method fails to converge. Also $m$ is the dimension of the matrix in LMI and $n$ is the dimension that $S$-$n$-$m$ lies.}} \label{tab:hmc_hnr}}
\end{table}

In Table~\ref{tab:hmc_hnr} we follow the cooling schedule 
described, after setting $T_0\approx R$ and sampling the first uniform
point with \billiard. We give the optimal solution as input and we
stop dropping $T$
when an error $\epsilon \leq 0.05$ is achieved.
Even in the case when the walk length is set equal to one, \hmcr still converges 
to the optimal solution. To the best of our knowledge, this
is the first time that a randomized algorithm, which is based on random
walks, is functional even when the walk length is set to one.
On the other hand, we set the walk length of \hitandrun\ $\OO(\sqrt{n})$ or $\OO(n)$ in our experiments. Notice that for the smaller walk length, its runtime decreases, but the method becomes unstable, as it sometimes fails to converge. For both cases its runtime is worse than that of \hmcr.

\section*{Acknowledgements}
ET is partially supported by ANR JCJC
GALOP (ANR-17-CE40-0009), the PGMO grant ALMA, and the PHC GRAPE.

\appendix
\section{Additional proofs}

To prove lemma \ref{lemma:gradient} we will need the following lemmas.

\begin{lemma}[Partial Derivative of Determinant]\label{lem:partialDerivativeDeterminant}		 
	\noindent Let $\bm{A}$ be a symmetric $m \times m$ matrix. Then
	
	\begin{displaymath}
	\frac{\partial \det\bm{A}}{\partial A_{ij}} = c_{ij}
	\end{displaymath}
	
	where $c_{ij}$ the cofactor of $A_{ij}$.
\end{lemma}

\begin{proof}
	From Laplace expansion:
	
	\begin{displaymath}
	\det\bm{A} = \sum\limits_{j=1}^m A_{ij}c_{ij}
	\end{displaymath}
	
	Notice that $c_{1j}, \cdots, c_{mj}$ are independent of $A_{ij}$, so we have
	
	\begin{displaymath}
	\frac{\partial \det\bm{A}}{\partial A_{ij}} = c_{ij}
	\end{displaymath}

\end{proof}

\begin{lemma}\label{lem:composition}
	Let $\lmi{x}=\bm{A}_0 + x_1 \bm{A}_1 + \cdots + x_n \bm{A}_n$. Then
	
	\begin{displaymath}
	\frac{\partial \det \lmi{x}}{\partial x_k} = \trace{\adjoint{\lmi{x}} \bm{A}_k)}
	\end{displaymath}
	
\end{lemma}

\begin{proof}	
	The function $\det\lmi{x}$ is the composition of $\det \bm{A}$ and $\bm{A} = \lmi{x}$, so from Lemma \ref{lem:partialDerivativeDeterminant} and the chain rule:
	
	\begin{displaymath}
	\frac{\partial \det\lmi{x}}{\partial x_k} = \sum\limits_{i=1}^m  \sum\limits_{j=1}^m \frac{\partial \det \bm{F}}{\partial \bm{F}_{ij}}\cdot 
	\frac{\partial \bm{F}_{ij}}{\partial x_k} = \sum\limits_{i=1}^m  \sum\limits_{j=1}^m c_{ij} A^k_{ij} = \trace{\adjoint{\lmi{x}} \bm{A}_k}
	\end{displaymath}

	where $ A^k_{ij}$ the $ij$-th element of matrix $\bm{A}_k$
\end{proof}

\begin{lemma}[Adjoint Matrix of $\bm{A}$]\label{adjoint}
	Let $\bm{A}$ be a $m \times m$ matrix of rank $r(\bm{A}) = m - 1$. Then
	
	\begin{displaymath}
	\adjoint{A} = \mu(\bm{A}) \frac{\bm{v}\bm{u}^\top}{\bm{u}^\top \bm{v}}
	\end{displaymath}
	
	where $\mu(\bm{A})$ is the product of the $m-1$ non-zero eigenvalues of $\bm{A}$, and $\bm{x}$ and $\bm{y}$ satisfy $\bm{A}\bm{v} =  \bm{A}^\top\bm{u} = \bm{0}$ (see chapter 3.2 in \cite{20c71247041945789c884f7b64057a9e}).
	
\end{lemma}


\section{Matrices of the Example}
The spectrahedron was randomly generated as in \cite{Dabbene2010ARC}.
Due to space considerations, the entries of the matrices are rounded to
the first decimal.

\begin{equation}
\bm{A_0} = \begin{bmatrix}
16.7&3.7&12.3&8.7&5.1&10.4\\
3.7&9.4&2.3&4&-2.3&-1\\
12.3&2.3&26.8&18.7&7.1&16.7\\
8.7&4&18.7&20&3.7&12.3\\
5.1&-2.3&7.1&3.7&6.1&5.4\\
10.4&-1&16.7&12.3&5.4&18.7
\end{bmatrix}
\end{equation}

\begin{equation}
\bm{A_1}=\begin{bmatrix}
0.5&-0.4&2.7&0&0&\\
-0.4&1.4&-0.2&0&0&0\\
2.7&-0.2&1.7&0&0&0\\
0&0&0&0.5&-0.4&2.7\\
0&0&0&-0.4&1.4&-0.2  \\
0&0&0&2.7&-0.2&1.7
\end{bmatrix} 
\end{equation}

\begin{equation}
\bm{A_2} =\begin{bmatrix}
2.6&-0.1&3&0&0&0\\
-0.1&1&-0.1&0&0&0\\
3&-0.1&-1&0&0&0\\
0&0&0&2.6&-0.1&3&  \\
0&0&0& -0.1&1&-0.1  \\
0&0&0&3&-0.1&-1
\end{bmatrix}
\end{equation}



\pagebreak
\bibliographystyle{abbrv}
\bibliography{sampling}  

\end{document}